\documentclass[draftclsnofoot,onecolumn,12pt]{IEEEtran}

\usepackage{amsfonts}
\usepackage{amsmath,amsthm}
\usepackage{amssymb}
\usepackage{graphicx}
\usepackage{caption}
\usepackage{pgf, tikz}
\providecommand{\U}[1]{\protect\rule{.1in}{.1in}}
\newtheorem{Theorem}{Theorem}
\newtheorem{Corollary}{Corollary}

\newtheorem{Example}{Example}
\newtheorem{Lemma}{Lemma}

\begin{document}

\title{Fine Asymptotics for Universal One-to-One Compression of Parametric Sources}
\author{Nematollah Iri and Oliver Kosut
\thanks{N. Iri and O. Kosut are with the School of Electrical, Computer and Energy Engineering, Arizona State University (Email:  niri1@asu.edu , okosut@asu.edu).   } \thanks{This paper was presented in part at the International Symposia on Information Theory in 2016 \cite{nemat2}.}\thanks{This material is based upon work supported by the National Science Foundation under Grant No. CCF-1422358.}}

\maketitle

\begin{abstract}
Universal source coding at short blocklengths is considered for an exponential family of distributions. The \emph{Type Size} code has previously been shown to be optimal up to the third-order rate for universal compression of all memoryless sources over finite alphabets. The Type Size code assigns sequences ordered based on their type class sizes to binary strings ordered lexicographically. To generalize this type class approach for parametric sources, a natural scheme is to define two sequences to be in the same type class if and only if they are equiprobable under any model in the parametric class. This natural approach, however, is shown to be suboptimal. A variation of the Type Size code is introduced, where type classes are defined based on neighborhoods of minimal sufficient statistics. Asymptotics of the overflow rate of this variation are derived and a converse result establishes its optimality up to the third-order term. These results are derived for parametric families of $i.i.d.$ sources as well as Markov sources.
\end{abstract}

\section{Introduction}
\IEEEPARstart{I}{n} the traditional source coding doctrine, performance of algorithms are characterized in the limit of large blocklengths. In some modern applications, however, data is continuously generated and updated, making them highly delay-sensitive. Therefore, it is vital to characterize the overheads associated with operation in the short blocklength regime.

To evaluate the performance of source coding for blocklengths at which the law of large numbers does not apply, we need a more refined metric than expected length. Thus, we use $\epsilon$-coding rate, the minimum rate such that the corresponding overflow probability is less than $\epsilon$. Fundamental limits of $\epsilon$-coding rate for fixed-to-variable lossless data compression in the non-universal setup are derived in \cite{verdulossless}, both for $i.i.d.$ as well as Markov sources. In most applications, however, the statistics of the source are unknown or arduous to estimate, especially at short blocklengths, where we have constraints on the available data for the inference task. In the universal setup, a class of models is given, however the true model in the class that generates the data is unknown. From an algorithmic angle, the aim of universal source coding is to propose a compression algorithm in which the encoding process is ignorant of the underlying unknown parameters, yet achieving the performance criteria.

Analysis of the finite blocklength behavior as well as fine asymptotics of universal source coding have been considered in \cite{oliver,oliver2,oliverArxiv,tanAsyEs} for the class of $i.i.d.$ sources, and in \cite{nemat} for the class of Markov sources. Similar to the aforementioned works, the universal source coding scheme in this paper compresses the whole file, so we relax the prefix condition \cite{szpan}. Imposing the prefix free condition, the $\epsilon$-coding rate of the Two Stage code \cite{oliver2,oliverArxiv} and that of the Bayes code \cite{saito,saitoISIT} are also considered in the literature.

The Type Size code (TS code) is introduced in \cite{oliver} for compression of the class of \emph{all} stationary memoryless sources, in which sequences are encoded in increasing order of type class size. It is shown that the resulting third-order term is $\frac{|\mathcal{X}|-3}{2}\log{n}$ bits, where $|\mathcal{X}|$ is the alphabet size. Its optimality is shown in \cite{oliver2}. Subsequently, a converse bound is derived in \cite{fekri} for one-to-one average minimax (and maximin) redundancy of memoryless sources, which consequently shows that the TS code is optimal up to $o(\log n)$ for universal one-to-one compression of \emph{all} memoryless sources, considering expected length as the performance metric \cite{fekri}. However, an achievable scheme for universal one-to-one compression of parametric sources with more \emph{structure} is not provided. Departing from average case analysis, we consider $\epsilon$-coding rate as the performance metric and provide an achievable scheme for compressing exponential families of distributions as the parametric class. Moreover, we provide a converse result, showing that our proposed scheme is optimal up to the third-order coding rate.

Type classes in \cite{oliver,nemat,fekri} are based on the empirical probability mass function (EPMF). In particular, two sequences are in the same (elementary) type class if they have the same EPMF. Elementary type classes do not exploit the inherited structure in the model class. To generalize the notion of a type to richer model classes, we define the \emph{point} type class as the set of sequences equiprobable under any model in the class. The size of the point type class structure is analyzed in \cite{merhav}. This natural characterization of type classes is based on the philosophy that the sequences with the same probability (under any model in the class) are ``\emph{indistinguishable}''. Such a philosophy has been employed before in the relevant applications, e.g. the universal simulation \cite{merhav} and the universal random number generation \cite{gadiel} problems. Perhaps surprisingly, we show that this natural approach is suboptimal for the universal source coding problem. In this paper, we characterize the structure of the type classes in a new fashion for the sake of optimally compressing exponential families of distributions. We refer to this new approach as \emph{quantized} types. We divide the convex hull of the set of minimal sufficient statistics into cuboids. Two sequences are in the same quantized type class if their minimal sufficient statistics belong to the same cuboid. Therefore, we show that \emph{approximate} indistinguishability leads to optimality for the source coding problem.

We consider fixed-to-variable length codes for a $d$-dimensional exponential family of distributions over a finite alphabet $\mathcal{X}$. For ease of exposition, we first assume, data generated by the unknown true model in this family is independent and identically distributed ($i.i.d.$). We subsequently extend the results to Markov data generation mechanisms. We provide performance guarantees for the Type Size code for these model classes. Using the Type Size code, we show that the minimal number of bits required to compress a length-$n$ sequence with probability $1-\epsilon$ is at most
\begin{equation}
\label{introResEq}
nH+\sigma\sqrt{n}Q^{-1}\left(\epsilon\right)+\left(\frac{d}{2}-1\right)\log{n}+\mathcal{O}\left(1\right)
\end{equation}
where $H$ and $\sigma^2$ are the entropy and varentropy of the underlying source respectively, $Q(\cdot)$ is the tail of the standard normal distribution and $d$ is the dimension of the model class. The first two terms in (\ref{introResEq}) are the same as the non-universal case \cite{verdulossless}, while the third-order $\log{n}$ term represents the cost of universality; for comparison, in the non-universal case the third-order term is $-\frac{1}{2}\log{n}$ \cite{verdulossless}. Precise bounds on the fourth-order $\mathcal{O}(1)$ term is beyond the scope of this paper. However, analyzing the fourth-order term is considered in the literature for the related source coding problems. For example, it is shown in \cite{szpan2008} that the fourth-order term is either a constant or has fluctuating behavior for average codelength of a binary memoryless source.

The rest of the paper is organized as follows. We introduce the exponential family, the finite-length lossless source coding problem and related definitions in Section \ref{sec::prelim}. In Section \ref{sec::TSC}, we describe quantized type classes and the variation of the TS code used in this paper. In Section \ref{sec::mainThm}, we present the main theorem of the paper, which characterizes the performance of the TS code using quantized type classes up to third order. We present preliminary results including a lemma bounding the size of a type class in Section \ref{sec::preRes}. We provide the proof of main theorem in Section \ref{sec::proofMain}. Extensions to the Markov case is considered in Section \ref{sec::ParMrk}. We show the suboptimality of the approach based on point type classes in Section \ref{sec::AltApprch}. We conclude in Section \ref{sec::conclusion}. A number of proofs are given in the appendices.

\section{Problem Statement}
\label{sec::prelim}
Let $\Theta$ be a compact subset of $\mathbb{R}^d$ with non-empty interior. Probability distributions in an exponential family can be expressed in the form \cite{merhav}
\begin{equation}
\label{pThetaEq}
p_{\theta}(x)=2^{\left\langle\theta,\boldsymbol{\tau}(x)\right\rangle - \psi(\theta)}
\end{equation}
where $\theta\in\Theta$ is the $d$-dimensional parameter vector, $\boldsymbol{\tau}(x): \mathcal{X}\rightarrow \mathbb{R}^d$ --- the crux of our parametric approach --- is the vector of sufficient statistics and $\psi(\theta)$ is the normalizing factor. Let the model class $\mathcal{P}=\left\{p_{\theta},\theta\in\Theta\right\}$, be the exponential family of distributions over the finite alphabet $\mathcal{X}=\left\{1,\cdots,|\mathcal{X}|\right\}$, parameterized by $\theta\in\Theta\subset \mathbb{R}^d$, where $d$ is the degrees of freedom in the minimal description of $p_{\theta}\in\mathcal{P}$ in the sense that no smaller dimensional family can capture the same model class. The degrees of freedom turns out to characterize the richness of the model class in our context. Compactness of $\Theta$ implies existence of a constant upper bound $\wp$ on the norm of the parameter vectors, namely $\|\theta\|\leq \wp$ for all $\theta\in\Theta$. We denote the (unknown) true model in force as $p_{\theta^*}$. $\mathbb{P}_{\theta}$, $\mathbb{E}_{\theta}$ and $\mathbb{V}_{\theta}$ denote probability, expectation and variance with respect to $p_{\theta}$, respectively. All logarithms are in base 2. Instead of introducing different indices for every new constant $C_1,C_2,...$, the same letter $C$ is used to denote different constants whose precise values are irrelevant.

From (\ref{pThetaEq}), the probability of a sequence $x^n$ drawn $i.i.d.$ from a model $p_{\theta}$ in the exponential family takes the form \cite{merhav}
\begin{align}
p_{\theta}(x^n) &=\prod_{i=1}^{n}p_{\theta}(x_i) \nonumber \\
                &= \prod_{i=1}^{n}2^{\big\langle\theta,{\boldsymbol{\tau}}(x_i)\big\rangle-\psi(\theta)} \nonumber \\
                &=2^{\left\{n\big[\big\langle\theta,\boldsymbol{\tau}(x^n)\big\rangle-\psi(\theta)\big]\right\}}\label{pNdimEq}
\end{align}
where
\begin{equation}
\label{tauXnDef}
\boldsymbol{\tau}(x^n)=\frac{\sum_{i=1}^{n}{\boldsymbol{\tau}(x_i)}}{n} \in \mathbb{R}^d
\end{equation}
is a minimal sufficient statistic \cite{merhav}. Note that $\boldsymbol{\tau}(x)$ and $\boldsymbol{\tau}(x^n)$ are distinguished based upon their arguments.

We consider a fixed-to-variable code that encodes an $n$-length sequence from the parametric source to a variable-length bit string via a coding function
\begin{equation*}
\phi:\mathcal{X}^n\rightarrow \{0,1\}^*=\{\emptyset,0,1,00,01,10,11,000,\cdots\}.
\end{equation*}
We do not make the assumption that the code is prefix-free. Let $l(\phi(x^n))$ be the number of bits in the compressed binary string when $x^n$ is the source sequence. We gauge the performance of algorithms through the $\epsilon$-coding rate at blocklength $n$ given by
\begin{equation*}
R_n(\epsilon,\phi, p_{\theta^*}):=\min\left\{\frac{k}{n}: \mathbb{P}_{\theta^*}\Big[l(\phi(X^n))\geq k\Big]\leq \epsilon\right\}.
\end{equation*}

\section{Type Size Code}
\label{sec::TSC}
For the class of all memoryless sources over a finite alphabet $\mathcal{X}$, the fixed-to-variable TS code is introduced in \cite{oliver}, which sorts sequences based on the size of the elementary type class from smallest to largest and then encodes sequences to variable-length bit-strings in this order. More precisely, define the support of a sequence as the set of observed symbols in it. The output of the encoder consists of a header that encodes the support of the sequence and a body that maps sequences to binary strings based on the size of their type class, among all sequences with the support set indicated in the header. That is, if two sequences $x^n$ and $y^n$ have the same support and $|T_{x^n}|\leq|T_{y^n}|$, then $l\left(\phi(x^n)\right)\leq l\left(\phi(y^n)\right)$, where $T_{x^n}$ is the type class of $x^n$.

We borrow the spirit of the TS code, yet our approach for parametric sources departs from that of \cite{oliver} in two ways
\begin{enumerate}
  \item Rather than defining type classes based on the EPMFs, we use quantized type classes, which are based on the neighborhoods of the minimal sufficient statistics.
  \item We omit the header encoding the support of the observed sequence. This header is unnecessary given the assumption that $\Theta$ is compact, because under this assumption, for any distribution in $\mathcal{P}$, each letter $x\in\mathcal{X}$ occurs with some probability bounded away from zero. Thus, all letters are likely to be observed for even moderate blocklengths.
\end{enumerate}
We first define quantized type classes for the purpose of compressing the exponential family. We cover the convex hull of the set of minimal sufficient statistics $\mathcal{T}=\text{conv}\left\{\boldsymbol{\tau}(x^n): x^n\in\mathcal{X}^n\right\}$, into $d$-dimensional cubic grids --- cuboids --- of side length $\frac{s}{n}$, where $s>0$ is a constant. The union of such disjoint cuboids should cover $\mathcal{T}$. The position of these cuboids is arbitrary, however once we cover the space, the covering is fixed throughout. We represent each $d$-dimensional cuboid by its geometrical \emph{center}. Denote $G(\boldsymbol{\tau}_0)$ as the cuboid with center $\boldsymbol{\tau}_0$, more precisely
\begin{equation}
\label{cuboidEq}
G(\boldsymbol{\tau}_0):= \left\{\boldsymbol{z}+\boldsymbol{\tau}_0 \in \mathbb{R}^d: -\frac{s}{2n}<z_i\leq \frac{s}{2n} \mbox{ for } 1\leq i \leq d \right\}
\end{equation}
where  $z_i$ is the $i$-th component of the $d$-dimensional vector $\boldsymbol{z}$.
Let $\boldsymbol{\tau}_c(x^n)$ be the center of the cuboid that contains $\boldsymbol{\tau}(x^n)$. Let us denote $\mathcal{T}_c$ as the set of cuboid centers, i.e., $\mathcal{T}_c=\left\{\boldsymbol{\tau}_c(x^n):x^n\in\mathcal{X}^n\right\}$.

We then define the quantized type class of $x^n$ as
\begin{equation}
\label{typeClassDefEq}
T_{x^n}:=\left\{y^n\in\mathcal{X}^n: \boldsymbol{\tau}(y^n)\in G\left(\boldsymbol{\tau}_c(x^n)\right)\right\}
\end{equation}
the set of all sequences $y^n$ with minimal sufficient statistic belonging to the very same cuboid containing the minimal sufficient statistic of $x^n$ (See Figure \ref{fig::TClassPar}).
\begin{figure}
\centering
\captionsetup{justification=centering}
\begin{tikzpicture}
\draw[step=2cm,color=gray] (0,0) grid (6,6);
\node at (3,4.3) {$\frac{s}{n}$};
\node at (2.5,4.2) {$\longleftarrow$};
\node at (3.5,4.2) {$\longrightarrow$};
\node at (6.2,6.2) {$\mathcal{T}$};
\node at (2.3,2.2) {$\bullet$};
\node at (2.5,2.5) {$\boldsymbol{\tau(x^n)}$};
\node at (3,3) { $\bullet$ } ;
\node at (3.2,3.3) {$\boldsymbol{\tau_c(x^n)}$};
\draw[step=2cm,color=red] (2,2) grid (4,4);
\node at (3,3.7) {{$\color{red}{G(\boldsymbol{\tau_c(x^n)})}$}};
\end{tikzpicture}
\caption{Type class structure for the exponential families}
\label{fig::TClassPar}
\end{figure}
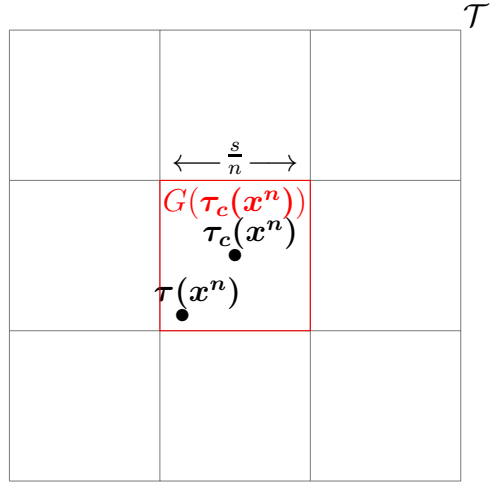

Since quantized type classes are represented by the cuboids and consequently the cuboid centers, we may interchangeably use $T_{\boldsymbol{\tau}_0}$ as the type class with corresponding cuboid center $\boldsymbol{\tau}_0$. Hence, $T_{\boldsymbol{\tau}_c(x^n)}$ is the same as $T_{x^n}$.

Two sequences within the given type class are indistinguishable from the coding perspective. The sequence indistinguishability introduced in this paper is reminiscent of the Balasubramanian's model indistinguishability \cite{bala}. In contrast to the sequence indistinguishability approach where the space of minimal sufficient statistics is partitioned into cuboids, in a model indistinguishability approach one may partition the source space. Asymptotics of the model indistinguishability approach is derived in \cite{JormaStCom}, where the maximum likelihood estimate is quantized to some precision by being the center of a cuboid. However, in their setup, the quantized code has the same logarithmic term as the maximum likelihood code with no quantization (See also \cite{JormaSNML}). For parametric TS code, the type class structure in \cite{merhav}, corresponds to the point type approach, wherein no quantization is done; i.e. $s=0$. In this limit, the size of the type class in \cite{merhav} depends on the dimension $d'$ of the derived lattice space \cite[Eq.A3]{merhav} rather than the model parameter dimension $d$. We return to this issue in Section \ref{sec::AltApprch}, wherein we show that using point types, the TS code achieves a third-order rate of $\left(\frac{d'}{2}-1\right)\log{n}$, which is not tight enough for our purposes due to the fact that $d'$ is in general larger than $d$.

As a direct consequence of our TS code construction, we have the following finite blocklength achievable bound; it constitutes a modification of Theorem 3 in \cite{oliver}.
\begin{Theorem}\cite{oliver}\label{finiteBlckThm}
For the TS code
\begin{equation}
R_n(\epsilon,\phi,p_{\theta^*})= \frac{1}{n}\left\lceil\log{M(\epsilon)}\right\rceil
\end{equation}
where
\begin{equation}
\label{mEq}
M(\epsilon)=\inf_{\gamma:\mathbb{P}_{\theta^*}\left(\frac{1}{n}\log{|T_{\boldsymbol{\tau}_c(X^n)}|}>\gamma\right)\leq \epsilon}\sum_{\substack{\boldsymbol{\tau}_c\in\mathcal{T}_c: \\ \frac{1}{n}\log{|T_{\boldsymbol{\tau}_c}|}\leq \gamma }}{{|T_{\boldsymbol{\tau}_c}|}}.
\end{equation}
\end{Theorem}

\section{Main Result}
Let $H(p_{\theta})=\mathbb{E}_{\theta}\left(\log{\frac{1}{p_{\theta}(X)}}\right)$ and $\sigma^2(p_{\theta})=\mathbb{V}_{\theta}\left(\log{\frac{1}{p_{\theta}(X)}}\right)$ be the entropy and the varentropy of $p_{\theta}$.
\label{sec::mainThm}
The following theorem exactly characterizes achievable $\epsilon$-rates up to third-order term, as well as asserting that this rate is achievable by the TS code.
\begin{Theorem}
\label{mainThm}
For any stationary memoryless exponential family of distributions parameterized by $\Theta$,
\begin{equation}
\inf_{\phi}\sup_{\theta\in\Theta}\left[R_n(\epsilon,\phi,{p_{\theta}})-H(p_{\theta})-\frac{\sigma(p_{\theta})}{\sqrt{n}}Q^{-1}(\epsilon)\right]=\left(\frac{d}{2}-1\right)\frac{\log{n}}{n}+\mathcal{O}\left(\frac{1}{n}\right) \label{mainEq}
\end{equation}
where the infimum is achieved by the TS code using quantized types.
\end{Theorem}

\begin{Example}
For the class of all $i.i.d.$ distributions $d=|\mathcal{X}|-1$, and Theorem \ref{mainThm} reduces to the result in \cite{oliver}.
\end{Example}

\section{Auxiliary Results}
\label{sec::preRes}
Define
\begin{equation}
\label{thetaHatEquation}
\hat{\theta}\left(\boldsymbol{\tau}\right)	=\underset{\theta\in\Theta}{\arg\max} \left(\langle\theta,\boldsymbol{\tau}\rangle-\psi(\theta)\right).
\end{equation}
Note that since the Hessian matrix of $\psi(\theta)$, $\boldsymbol{\nabla}^2\left(\psi(\theta)\right)=\text{Cov}_{\theta}\left(\boldsymbol{\tau}(X)\right)$ is positive definite, the log-likelihood function is strictly concave and hence the maximum likelihood $\hat{\theta}(\boldsymbol{\tau})$ is unique.

For notational convenience, we may omit the dependencies on $\boldsymbol{\tau}$ and $\boldsymbol{\tau}_c$ in $\hat{\theta}\left(\boldsymbol{\tau}(x^n)\right)$ and $\hat{\theta}\left(\boldsymbol{\tau}_c(x^n)\right)$, and simply denote them by $\hat{\theta}(x^n)$ and $\hat{\theta}_c(x^n)$, respectively.

The next lemma provides tight upper and lower bounds on the type class size. Beside its exclusive bearing, it is a main component of the achievability proof.
\begin{Lemma}[Type Class Size]
\label{TypeClSizeLem}
Let $\kappa=\wp\frac{\sqrt{d}}{2}$. For large enough $n$, the size of the type class of $x^n$ is bounded as
\begin{equation*}
r(x^n)-2\kappa s+C' \leq \log{|T_{x^n}|} \leq r(x^n)+2\kappa s+C
\end{equation*}
where
\begin{equation*}
r(x^n)=-\log{p_{\hat{\theta}_c(x^n)}(x^n)}-\frac{d}{2}\log{n}+d\log{s}
\end{equation*}
is the common part of the upper and lower bounds and $C,C'$ are constants independent of $n$.
\end{Lemma}
\begin{proof}
For notational convenience, when it is clear from the context, we may suppress the arguments in $\boldsymbol{\tau}_c(x^n)$ and $G(\boldsymbol{\tau}_c(x^n))$ and denote them simply as $\boldsymbol{\tau}_c$ and $G(\boldsymbol{\tau}_c)$.

Motivated by \cite[Eq. A2]{merhav}, we bound $|T_{x^n}|$ as follows:
\begin{equation}
\displaystyle
\frac{{\mathbb{P}}_{\hat{\theta}_c(x^n)}\left\{\boldsymbol{\tau}(X^n)\in G\left(\boldsymbol{\tau}_c(x^n)\right)\right\}}{{\displaystyle\max_{\substack{y^n:\\ \boldsymbol{\tau}(y^n)\in G\left(\boldsymbol{\tau}_c(x^n)\right)}}}{\mathbb{P}}_{\hat{\theta}_c(x^n)}(y^n)} \leq |T_{x^n}|  \leq \frac{{\mathbb{P}}_{\hat{\theta}_c(x^n)}\left\{\boldsymbol{\tau}(X^n)\in G\left(\boldsymbol{\tau}_c(x^n)\right)\right\}}{{\displaystyle\min_{\substack{y^n:\\ \boldsymbol{\tau}(y^n)\in G\left(\boldsymbol{\tau}_c(x^n)\right)}}{\mathbb{P}}_{\hat{\theta}_c(x^n)}(y^n)}}. \label{TypeMainEq}
\end{equation}
Let \begin{math}
nG(\boldsymbol{\tau}_c)=\left\{n\textbf{z}:\textbf{z}\in G(\boldsymbol{\tau}_c)\right\}.
\end{math}
It is clear that
\begin{equation*}
\mathbb{P}_{\hat{\theta}_c(x^n)}\left\{\boldsymbol{\tau}(X^n)\in G(\boldsymbol{\tau}_c)\right\}=\mathbb{P}_{\hat{\theta}_c(x^n)}\left\{n\boldsymbol{\tau}(X^n)\in nG(\boldsymbol{\tau}_c)\right\}.
\end{equation*}
Exploiting the result in \cite[Corollary 1]{stone}, we have
\begin{equation}
\mathbb{P}_{\hat{\theta}_c(x^n)}\left\{n\boldsymbol{\tau}(X^n)\in nG(\boldsymbol{\tau}_c)\right\} =\frac{s^d}{\left(2\pi n\right)^{\frac{d}{2}}|\boldsymbol{\Sigma}|^{\frac{1}{2}}}e^{-\frac{\left(n\boldsymbol{\tau}_c-n\boldsymbol{\mu}_c\right)\cdot \boldsymbol{\Sigma}^{-1}\cdot \left(n\boldsymbol{\tau}_c-n\boldsymbol{\mu}_c\right)}{2n}} +o\left(n^{-\frac{d}{2}}\right)\label{stoneEq}
\end{equation}
where $\boldsymbol{\mu}_c$ and $\boldsymbol{\Sigma}$ are the mean and the covariance (resp.) of $\boldsymbol{\tau}(X)$ under $\hat{\theta}_c(x^n)$. To proceed, we show that $\boldsymbol{\mu}_c=\boldsymbol{\tau}_c$. We have
\begin{align*}
\hat{\theta}_c(x^n)&=\underset{\theta\in\Theta}{\arg\min}\: \Big(D(p_{\hat{\theta}_c(x^n)}\|p_{\theta})+H(p_{\hat{\theta}_c(x^n)})\Big) \nonumber \\
							&=\underset{\theta\in\Theta}{\arg\max}\: \mathbb{E}_{\hat{\theta}_c(x^n)}\Big(\log{p_{\theta}(X)}\Big) \nonumber \\
							&=\underset{\theta\in\Theta}{\arg\max}\: \mathbb{E}_{\hat{\theta}_c(x^n)}\Big(\langle\theta,\boldsymbol{\tau}(X)\rangle-\psi(\theta)\Big) \nonumber \\
							&=\underset{\theta\in\Theta}{\arg\max}\: \langle\theta,\boldsymbol{\mu}_c\rangle-\psi(\theta).
\end{align*}
That is, $\hat{\theta}_c(x^n)$ is the maximum likelihood estimate for $\boldsymbol{\mu}_c$ and (by definition (\ref{thetaHatEquation})) $\boldsymbol{\tau}_c$. However, in order to be the maximum likelihood estimate, it must be that the derivative of the log-likelihood function is 0, hence $\nabla \psi(\hat{\theta}_c(x^n))=\boldsymbol{\mu}_c$ and $\nabla\psi(\hat{\theta}_c(x^n))=\boldsymbol{\tau}_c$. Therefore $\boldsymbol{\mu}_c$ and $\boldsymbol{\tau}_c$ are equal. Due to (\ref{stoneEq}) and $\boldsymbol{\mu}_c=\boldsymbol{\tau}_c$, there exist constants $C,C'$ such that, for large enough $n$,
\begin{equation}
d\log s -\frac{d}{2}\log n +C'\leq \log p_{\hat{\theta}_c(x^n)}\{\boldsymbol{\tau}(X^n)\in G(\boldsymbol{\tau}_c)\} \leq  d\log s -\frac{d}{2}\log n +C.\label{logPthetaEq}
\end{equation}
On the other hand
\begin{equation*}
\log p_{\hat{\theta}_c(x^n)}(x^n) =n \left[\langle\hat{\theta}_c(x^n),\boldsymbol{\tau}(x^n)\rangle-\psi\left(\hat{\theta}_c(x^n)\right)\right].
\end{equation*}
Therefore
\begin{equation}
\label{eqTwo}
\max_{\substack{y^n:\\ \boldsymbol{\tau}(y^n)\in G(\boldsymbol{\tau}_c(x^n))}}\log{p_{\hat{\theta}_c(x^n)}(y^n)}\leq \log p_{\hat{\theta}_c(x^n)}(x^n)+2\kappa s
\end{equation}
and
\begin{equation}
\label{eqOne}
\min_{\substack{y^n:\\ \boldsymbol{\tau}(y^n)\in G(\boldsymbol{\tau}_c(x^n))}}\log{p_{\hat{\theta}_c(x^n)}(y^n)}\geq \log p_{\hat{\theta}_c(x^n)}(x^n)-2\kappa s
\end{equation}
where we used  $\|\hat{\theta}_c(x^n)\|\leq \wp$ and the fact that if $\boldsymbol{\tau}(x^n)$ and $\boldsymbol{\tau}(y^n)$ belong to the same cuboid, then $\|\boldsymbol{\tau}(x^n)-\boldsymbol{\tau}(y^n)\|<\frac{s\sqrt{d}}{n}$. Plugging (\ref{logPthetaEq},\ref{eqTwo},\ref{eqOne}) in (\ref{TypeMainEq}), the lemma follows.
\end{proof}
\begin{Corollary}
\label{TypeCorollary}
For large enough $n$, the size of the type class of $x^n$ with corresponding cuboid center $\boldsymbol{\tau}_c$ is bounded as
\begin{equation*}
nf(\boldsymbol{\tau}_c)-6\kappa s -C''\leq \log{|T_{\boldsymbol{\tau}_c}|} \leq nf(\boldsymbol{\tau}_c)
\end{equation*}
where, $C''=C-C'$ and
\begin{equation}
f(\boldsymbol{\tau})=-\langle\hat{\theta}(\boldsymbol{\tau}),\boldsymbol{\tau}\rangle + \psi\left(\hat{\theta}(\boldsymbol{\tau})\right)-\frac{d}{2n}\log{n}+\frac{d\log{s}}{n} +\frac{3\kappa s}{n}+\frac{C}{n}.\label{fEqUpp2}
\end{equation}
\end{Corollary}

We appeal to the following normal approximation result in order to bound the CDF of the type class size (in the achievability proof) and further CDF of the mixture distribution (in the converse proof) with that of the normal distribution.
\begin{Lemma}[Asymptotic Normality of Information]\label{maxLikeBerr}
Fix a positive constant $\alpha$. For a stationary memoryless source, there exists a finite positive constant $A$, such that for all $n\geq 1$ and $z$ such that $|z|\leq \alpha$,
\begin{equation}
\left|\mathbb{P}_{\theta^*}\left\{\frac{-\log{p_{\hat{\theta}(X^n)}(X^n)}-nH}{\sqrt{n}\sigma}>z\right\}-Q(z)\right|\leq \frac{A}{\sqrt{n}}
\end{equation}
where $H:=H(p_{\theta^*})$ and $\sigma^2:=\sigma^2(p_{\theta^*})$, are the entropy and varentropy of the true model $p_{\theta^*}$, respectively.
\end{Lemma}
\begin{proof}
See Appendix \ref{app::maxLikeBerr}
\end{proof}

The following lemma provides a guarantee in approximation of $p_{\hat{\theta}(x^n)}(x^n)$ with $p_{\hat{\theta}_{c}(x^n)}(x^n)$, which allows us to use the Lemma \ref{maxLikeBerr} in the achievability proof.
\begin{Lemma}[Maximum Likelihood Approximation]
\label{apprxLemm}
Let $\kappa$ be defined as in Lemma \ref{TypeClSizeLem}. We have
\begin{equation*}
\log{p_{\hat{\theta}(x^n)}(x^n)}-\log{p_{\hat{\theta}_c(x^n)}(x^n)}\leq 2\kappa s.
\end{equation*}
\end{Lemma}
\begin{proof}
See Appendix \ref{app::apprxLemm}.
\end{proof}

We need the following machinery lemmas for the achievability proof.
\begin{Lemma}
\label{fIsLipschitz}
There exists a Lipschitz constant $K_0$ independent of $n$, so that for any minimal sufficient statistics $\boldsymbol{\tau}_1$ and $\boldsymbol{\tau}_2$,
\begin{equation}
|f(\boldsymbol{\tau}_1)-f(\boldsymbol{\tau}_2)|\leq K_0 \|\boldsymbol{\tau}_1-\boldsymbol{\tau}_2\|.
\end{equation}
\end{Lemma}
\begin{proof}
See Appendix \ref{app:fLipschitz}.
\end{proof}

Let $\omega=\frac{\log{|\mathcal{X}|}-H}{5}$. Without loss of generality, we may assume that the true model is non-uniform distribution, otherwise TS code (like any other rational code) is obviously optimal. Therefore, $\omega>0$. Let $0\leq \lambda< H+\omega$, and $\rho(\lambda)=\text{Vol}\left\{\boldsymbol{\tau}:f(\boldsymbol{\tau})\leq \lambda\right\}$ be the volume of the sub-level sets.
\begin{Lemma}
\label{rhoIsLipschitz}
There exists a Lipschitz constant $K_1$ so that for all $0\leq a,b< H+\omega$,
\begin{equation*}
|\rho(a)-\rho(b)|\leq K_1|a-b|.
\end{equation*}
\end{Lemma}
\begin{proof}
See Appendix \ref{app:hLipschitz}.
\end{proof}

For our converse proof, we will need the regular value theorem \cite[Prop. 2.3.2]{balasko} from manifold theory (see also  \cite[Theorem 9]{robbin}), stated as follows.
\begin{Theorem}
\label{thetaZeroDim}
Let $M$ and $N$ be smooth manifolds of dimensions $m_1,m_2$ with $m_1\geq m_2$. Let $\eta_0:M\longrightarrow N$ and $b\in N$ be such that for any $a\in\eta_0^{-1}(b)$, the Jacobian matrix of $\eta_0$ at $a$ is a surjective map from $M$ to $N$. Then, $\eta_0^{-1}(b)$ is a $(m_1-m_2)$-dimensional manifold.
\end{Theorem}

We have the following Laplace's approximation theorem for the integral of manifolds. We refer the reader to \cite[Chap. 9, Th. 3]{wong} for a detailed proof. In the converse proof, we use the Laplace's approximation to bound the self information of the mixture distribution.
\begin{Theorem}[Laplace's Approximation]\label{laplaceTheorem}\cite{oliver2}
Let $D$ be a $\tilde{d}-$dimensional differentiable manifold embedded in $\mathbb{R}^m$ and $\eta_1(\cdot)$ and $\eta_2(\cdot)$ be functions that are infinitely differentiable on $D$. Let
\begin{equation}
\label{lapint}
Z(n)=\int_{D}{\eta_2(x)e^{-n\eta_1(x)}dx}
\end{equation}
Assume that: (i) the integral in (\ref{lapint}) converges absolutely for all $n\geq n_0$; (ii) there exists a point $x^*$ in the interior of $D$ such that for every $\epsilon>0$, $\xi(\epsilon)>0$ where
\begin{equation*}
\xi(\epsilon)=\inf\left\{\eta_1(x)-\eta_1(x^*): x\in D \mbox{ and } |x-x^*|\geq \epsilon\right\}
\end{equation*}
and (iii) the Hessian matrix $\mathcal{E}=\left(\frac{\partial^2 \eta_1(x)}{\partial{x}_i\partial x_j}\right) \Big| _{x=x^*}$ is positive definite. Let $F\in \mathbb{R}^{m\times \tilde{d}}$ be an orthonormal basis for the tangent space to $D$ at $x^*$. Then
\begin{equation*}
Z(n)=e^{-n\eta_1(x^*)}\left(\frac{2\pi}{n}\right)^{\frac{\tilde{d}}{2}}\eta_2(x^*)\left|F^T\mathcal{E}F\right|^{-\frac{1}{2}}\left(1+\mathcal{O}\left(\frac{1}{n}\right)\right)
\end{equation*}
\end{Theorem}

\section{Proof of Theorem \ref{mainThm}}
\label{sec::proofMain}
\subsection{Achievability}
\label{subsec::Achiev}
In this subsection we bound the third-order coding rate of the quantized implementation of the TS code. We continue from the finite blocklength result in Theorem \ref{finiteBlckThm}, and evaluate its asymptotic performance.

For the constants $C$ and $A$ in Lemmas \ref{TypeClSizeLem} and \ref{maxLikeBerr}, let
\begin{equation}
\label{gammaEq}
\gamma = H+\frac{\sigma}{\sqrt{n}}Q^{-1}\Big(\epsilon-\frac{A}{\sqrt{n}}\Big)-\frac{d}{2n}\log{n} +\frac{d}{n}\log{s}+\frac{4\kappa s}{n}+\frac{C}{n}.
\end{equation}
Denote
\begin{align}
p_{\gamma}&:=\mathbb{P}_{\theta^*}\Big[\log{|T_{X^n}|}>n\gamma\Big] \label{pGammaFirst}\\
&=\mathbb{P}_{\theta^*}\Big[\log{|T_{\boldsymbol{\tau}_c(X^n)}|}>n\gamma\Big].\nonumber
\end{align}
We have
\begin{align}
p_{\gamma} &\leq \mathbb{P}_{\theta^*}\Big[-\log{{p_{\hat{\theta}_c(x^n)}(X^n)}}>nH +\sqrt{n}\sigma Q^{-1}\Big(\epsilon-\frac{A}{\sqrt{n}}\Big) +2\kappa s\Big] \label{aEq} \\		 &\leq\mathbb{P}_{\theta^*}\Big[\frac{-\log{{p_{\hat{\theta}(x^n)}(X^n)}}-nH}{\sqrt{n}\sigma}>Q^{-1}\Big(\epsilon-\frac{A}{\sqrt{n}}\Big)  \Big] \label{bEq} \\				
&\leq Q\Big(Q^{-1}\big(\epsilon-\frac{A}{\sqrt{n}}\big)\Big)+\frac{A}{\sqrt{n}} \label{dEq} \\
&=\epsilon \nonumber
\end{align}
where (\ref{aEq}) follows from Lemma \ref{TypeClSizeLem} and (\ref{gammaEq}), (\ref{bEq}) is from Lemma \ref{apprxLemm}, and (\ref{dEq}) is a consequence of Lemma \ref{maxLikeBerr}. Since for $\gamma$ in (\ref{gammaEq}), we have $p_{\gamma}\leq \epsilon$, therefore it satisfies the constraint of (\ref{mEq}). We can therefore, bound $M(\epsilon)$ defined in (\ref{mEq}), with this choice of $\gamma$. Fixing $\Delta=\frac{1}{n}$, we have
\begin{align}
M(\epsilon) &\leq \sum_{\substack{\boldsymbol{\tau}_c\in\mathcal{T}_c:\\ \frac{1}{n}\log{|T_{\boldsymbol{\tau}_c}|}\leq \gamma}}{|T_{\boldsymbol{\tau}_c}|} \nonumber \\
            &\leq \sum_{\substack{\boldsymbol{\tau}_c\in\mathcal{T}_c:\\ f(\boldsymbol{\tau}_c)-\frac{6\kappa s+C''}{n}\leq \gamma}}{2^{nf(\boldsymbol{\tau}_c)}} \label{useBoundLem} \\
            &= \sum_{i=0}^{\infty}\sum_{\substack{\boldsymbol{\tau}_c\in\mathcal{T}_c:\\ f(\boldsymbol{\tau}_c)\in\mathcal{A}_i }} {2^{nf(\boldsymbol{\tau}_c)}} \nonumber \\
						&\leq \sum_{i=0}^{\infty}\left|\left\{\boldsymbol{\tau}_c\in\mathcal{T}_c:f(\boldsymbol{\tau}_c)\in\mathcal{A}_i\right\}\right| \cdot  2^{n\gamma+6\kappa s+C''-ni\Delta} \label{mEpsilon}
\end{align}
where (\ref{useBoundLem}) follows from Corollary \ref{TypeCorollary} and $\mathcal{A}_i=\left(\gamma+\frac{6\kappa s+C''}{n}-(i+1)\Delta, \gamma+\frac{6\kappa s+C''}{n}-i\Delta\right]$. The rest of the proof is similar to \cite{oliver}, however we continue the proof for completeness. We have

\begin{align}
\left|\left\{\boldsymbol{\tau}_c\in\mathcal{T}_c:f(\boldsymbol{\tau}_c)\in\mathcal{A}_i\right\}\right|&=\sum_{\substack{\boldsymbol{\tau}_c\in\mathcal{T}_c:\\f(\boldsymbol{\tau}_c)\in\mathcal{A}_i}}\frac{\text{Vol}\left(G(\boldsymbol{\tau}_c)\right)}{\left(\frac{s}{n}\right)^d} \label{sumOne} \\
&= \frac{1}{\left(\frac{s}{n}\right)^d}{\text{Vol}\left(\bigcup_{\substack{\boldsymbol{\tau}_c\in\mathcal{T}_c:\\f(\boldsymbol{\tau}_c)\in\mathcal{A}_i}}G(\boldsymbol{\tau}_c)\right)} \label{aaEq2} \\
&\leq \frac{1}{\left(\frac{s}{n}\right)^d}{\text{Vol}\left(\bigcup_{\boldsymbol{\tau}\in\mathcal{T}:f(\boldsymbol{\tau})\in\mathcal{A}_i}G(\boldsymbol{\tau})\right)} \nonumber
\end{align}
where (\ref{sumOne}) results from $\text{Vol}\left(G(\boldsymbol{\tau}_c)\right)=\left(\frac{s}{n}\right)^d$, (\ref{aaEq2}) follows from disjointness of the cuboids. If $\boldsymbol{\tau}\in G(\boldsymbol{\tau}_c)$, then $\|\boldsymbol{\tau}-\boldsymbol{\tau}_c\|\leq\frac{s\sqrt{d}}{2n}$ and consequently by Lemma \ref{fIsLipschitz}
\begin{equation}
|f(\boldsymbol{\tau})-f(\boldsymbol{\tau}_c)|\leq K_0\cdot \frac{s\sqrt{d}}{2n}  := K_2\frac{s}{n} \label{ffLipDist}
\end{equation}
where $K_2=K_0\frac{\sqrt{d}}{2}$. Therefore, for $a=\gamma+\frac{6\kappa s+C''}{n}-(i+1)\Delta$,
\begin{align}
\left|\left\{\boldsymbol{\tau}_c\in\mathcal{T}_c:f(\boldsymbol{\tau}_c)\in \mathcal{A}_i\right\}\right|
&\leq \frac{1}{(\frac{s}{n})^d}\cdot \text{Vol}\Big(\bigcup_{ a<f(\boldsymbol{\tau})\leq a+\Delta}{G(\boldsymbol{\tau})}\Big) \nonumber \\
&\leq \frac{1}{(\frac{s}{n})^d}\text{Vol}\left(\left\{\boldsymbol{\tau}:f(\boldsymbol{\tau})\in \left(a-K_2\frac{s}{n},a+\Delta+K_2\frac{s}{n}\right]\right\}\right) \label{bistoshish} \\
&= \frac{1}{\left(\frac{s}{n}\right)^d}\left[\rho\left(a+\Delta+K_2\frac{s}{n}\right)-\rho\left(a-K_2\frac{s}{n}\right)\right] \label{breakPoint}
\end{align}
where (\ref{bistoshish}) is from (\ref{ffLipDist}). In order to continue from (\ref{breakPoint}), recall $\omega=\frac{\log{|\mathcal{X}|}-H}{5}$. Observe that by (\ref{gammaEq}) , $a+K_2\frac{s}{n}+\Delta\leq H+\frac{C}{\sqrt{n}}$, for a positive constant $C$. Since $\omega>0$, $H+\frac{C}{\sqrt{n}}< H+\omega$ for large enough $n$.  Similar argument shows that $0\leq a-K_2\frac{s}{n}<H+\omega$. Therefore boundary conditions of Lemma \ref{rhoIsLipschitz} are satisfied. Continuing from (\ref{breakPoint}) and using Lemma \ref{rhoIsLipschitz}, we then have
\begin{equation}
\left|\left\{\boldsymbol{\tau}_c\in\mathcal{T}_c:f(\boldsymbol{\tau}_c)\in \mathcal{A}_i\right\}\right|\leq \frac{K_1}{\left(\frac{s}{n}\right)^d}\cdot \left[\Delta+2K_2\frac{s}{n}\right]. \label{sizeTau}
\end{equation}
Applying (\ref{sizeTau}) to (\ref{mEpsilon}), we obtain
\begin{align*}
M(\epsilon) &\leq \sum_{i=0}^{\infty}{\frac{K_1}{(\frac{s}{n})^d}\cdot \left[\Delta+2K_2\frac{s}{n}\right]\cdot 2^{n\gamma+6\kappa s+C''-ni\Delta}} \nonumber \\
            &= \frac{n^d}{s^d}\cdot \left[\Delta+2K_2\frac{s}{n}\right]\cdot 2^{n\gamma+6\kappa s+C''}\cdot \frac{K_1}{1-2^{-n\Delta}}.
\end{align*}
From (\ref{gammaEq}) and since $s>0$ is a constant and $\Delta=\frac{1}{n}$, we obtain
\begin{equation*}
\log {M(\epsilon)}\leq nH+\sigma\sqrt{n}Q^{-1}(\epsilon)+\left(\frac{d}{2}-1\right)\log{n}+\mathcal{O}(1).
\end{equation*}
\subsection{Converse}
For a parameter vector $\theta\in\Theta$, define $J(\theta)=nH(p_{\theta})+\sigma(p_{\theta})\sqrt{n}Q^{-1}(\epsilon)$. We first rewrite the entropy function as follows:
\begin{align}
H(p_{\theta}) &= -\sum_{x\in\mathcal{X}}{p_{\theta}(x)\log{p_{\theta}}(x)} \nonumber \\
              &= -\sum_{x\in\mathcal{X}}{p_{\theta}(x)\left(\langle\theta,\boldsymbol{\tau}(x)\rangle-\psi(\theta)\right)} \label{defThetaP} \\
              &= -\langle\theta,\mathbb{E}_{\theta}(\boldsymbol{\tau}(x))\rangle+\psi(\theta) \nonumber \\
              &= -\langle\theta,\nabla\psi(\theta)\rangle+\psi(\theta) \label{fromFundEq}
\end{align}
where (\ref{defThetaP}) is from (\ref{pThetaEq}) and (\ref{fromFundEq}) is from $\mathbb{E}_{\theta}(\boldsymbol{\tau}(x))=\nabla\psi(\theta)$ \cite{jordan}. Taking derivative of (\ref{fromFundEq}) with respect to $\theta$, we obtain
\begin{equation}
\label{findZeros}
\nabla H(p_{\theta}) = -\theta\nabla^2\psi(\theta).
\end{equation}
Since $\nabla^2\psi(\theta) = \text{Cov}(\boldsymbol{\tau}(X))$ is positive definite, (\ref{findZeros}) vanishes only at the uniform distribution $\theta_{u}=(0,\cdots,0)$. Since $\Theta$ has nonempty interior, let $\theta_0$ be a point in the interior of $\Theta$ with $J(\theta_0)\neq J(\theta_u)$. Define
\begin{equation*}
\Theta_0:=\left\{\theta\in\Theta:J(\theta)=J(\theta_0)\right\}.
\end{equation*}
As $\theta_u\notin\Theta_0$, $\nabla H(p_{\theta})$ is nonzero for all parameters $\theta\in\Theta_0$. Therefore, for large enough $n$, $\nabla J(\theta)$ is also nonzero for all $\theta\in\Theta_0$. Hence, the Jacobian of $J(\cdot)$ at any point in the set $J^{-1}(J(\theta_0))$ is a surjective map from $\Theta_0$ to $\mathbb{R}$. Theorem \ref{thetaZeroDim} then implies that $\Theta_0$ is a $(d-1)$-dimensional manifold.

In order to prove the converse, it suffices to show that
\begin{equation*}
\sup_{\theta\in\Theta_0}R_n(\epsilon,\phi,p_{\theta})\geq \frac{J(\theta_0)}{n}+\left(\frac{d}{2}-1\right)\frac{\log{n}}{n}-\mathcal{O}\left(\frac{1}{n}\right).
\end{equation*}

Let $\overline{p}(x^n)$ be the mixture distribution with uniform prior among $n$-length $i.i.d.$ distributions with marginals parametrized by $\Theta_0$, i.e.
\begin{equation}
\label{mixtureEq}
\overline{p}(x^n)=\frac{1}{\text{Vol}(\Theta_0)}\int_{\theta\in\Theta_0}p_{\theta}(x^n)d\theta
\end{equation}
where $\text{Vol($\cdot$)}$ is the $d$-dimensional volume. For any $\gamma>0$, applying Theorem 3 in \cite{oliver2} gives
\begin{equation}
\label{epsilon2tauEq}
\epsilon+2^{-\gamma}\geq \inf_{\theta\in\Theta_0}\mathbb{P}_{\theta}\left(\iota_{\overline{p}}(X^n)\geq k+\gamma\right)
\end{equation}
where $\iota_{\overline{p}}(X^n):=-\log{\overline{p}(X^n)}$ is the self information of the mixture distribution. We then provide a lower bound for the self information. We may rewrite (\ref{mixtureEq}) as
\begin{align*}
\overline{p}(x^n)=\frac{1}{\text{Vol}(\Theta_0)}\int_{\theta\in\Theta_0}2^{-g(\theta)}d\theta
\end{align*}
where $g(\theta):=-\log{p_{\theta}(x^n)}$. Since $\Theta_0$ is a $(d-1)$-dimensional manifold, application of the Laplace's approximation of integrals (Theorem \ref{laplaceTheorem}) yields
\begin{equation}
\label{overlineEq}
\overline{p}(x^n)=\frac{1}{\text{Vol}(\Theta_0)}2^{-g(\hat{\theta})}\left(\frac{2\pi}{n}\right)^{\frac{d-1}{2}}\left|F^T\mathcal{E}F\right|^{-\frac{1}{2}}\left(1+\mathcal{O}\left(\frac{1}{n}\right)\right)
\end{equation}
where $\hat{\theta}:=\hat{\theta}(x^n)$ is the maximum likelihood estimate of $\theta$ for $x^n$. Continuing from (\ref{epsilon2tauEq}) for a constant $C>0$, we obtain
\begin{align}
&\epsilon+2^{-\gamma}\nonumber \\
                    &\geq \inf_{\theta\in\Theta_0}\mathbb{P}_{\theta}\left(\iota_{\overline{p}}(X^n)\geq k+\gamma \right) \nonumber \\
                    &\geq\inf_{\theta\in\Theta_0}\mathbb{P}_{\theta}\left(-\log{p_{\hat{\theta}}(X^n)}+\frac{d-1}{2}\log{n}+C\geq k+\gamma\right) \label{converseAEQ}\\
                    &=\inf_{\theta\in\Theta_0}\mathbb{P}_{\theta}\Bigg(\frac{-\log{p_{\hat{\theta}}(X^n)}-nH(p_{\theta})}{\sqrt{n}\sigma} \geq\frac{k+\gamma-\frac{d-1}{2}\log{n}-C-nH(p_{\theta})}{\sqrt{n}\sigma}\Bigg) \nonumber\\
                    &\geq Q\left(\frac{k+\gamma-\frac{d-1}{2}\log{n}-C-nH(p_{\theta})}{\sqrt{n}\sigma}\right) -\frac{A}{\sqrt{n}} \label{converseBerryEsseen}
\end{align}
where (\ref{converseAEQ}) is due to (\ref{overlineEq}) and the definition of $g(\cdot)$, while (\ref{converseBerryEsseen}) is from Lemma \ref{maxLikeBerr}.
Setting $\gamma=\frac{1}{2}\log{n}$ and rearranging gives
\begin{equation*}
\frac{k}{n}\geq \inf_{\theta\in\Theta_0}H(p_{\theta})+\frac{\sigma(p_{\theta})}{\sqrt{n}}Q^{-1}\left(\epsilon+\frac{A+1}{\sqrt{n}}\right)+\left(\frac{d}{2}-1\right)\frac{\log{n}}{n}+\frac{C}{n}.
\end{equation*}
Recalling that $H(p_{\theta})+\frac{\sigma(p_{\theta})}{\sqrt{n}}Q^{-1}(\epsilon)$ is fixed at $\frac{J(\theta_0)}{n}$ for all $\theta\in\Theta_0$ and that $\frac{k}{n}=\max_{\theta\in\Theta_0}R_n(\epsilon,\phi,p_{\theta})$, theorem follows.

\section{Parametric Markov Class}
\label{sec::ParMrk}
We now consider extensions to the class of parametric Markov models. Let $\mathcal{M}$ be the exponential family of first-order, stationary, irreducible and aperiodic Markov sources, parametrized by a $d$-dimensional parameter vector $\theta\in\Theta_{\mathcal{M}}\subset\mathbb{R}^d$. Transition probabilities of the distribution $p_{\theta}\in\mathcal{M}$ has the following exponential structure
\begin{equation}
\label{parMrkBaseEq}
p_{\theta}(x_{i}|x_{i-1})=2^{\langle \theta, \boldsymbol{\tau}(x_{i-1},x_{i})\rangle -\psi(\theta)}
\end{equation}
where $\boldsymbol{\tau}:\mathcal{X}\times\mathcal{X}\to \mathbb{R}$ is the vector of sufficient statistics.

Similar to $\cite{merhavVFvsFV}$, we assume that the initial source symbol $x_0$ is fixed and known to both the encoder and the decoder. From (\ref{parMrkBaseEq}), the probability of a sequence $x^n$ drawn according to the first-order Markov source $p_{\theta}\in\mathcal{M}$ in the exponential family takes the form
\begin{align*}
p_{\theta}(x^n)&=\prod_{i=1}^{n}{p_{\theta}\left(x_i|x_{i-1}\right)} \\
               &= \prod_{i=1}^{n}{2^{\left\langle\theta,\boldsymbol{\tau}(x_{i-1},x_i)\right\rangle-\psi(\theta)}} \\
               &=2^{n\left[\left\langle \theta,\boldsymbol{\tau}(x^n)\right\rangle - \psi\left(\theta\right)\right]}
\end{align*}
where $\boldsymbol{\tau}(x^n)=\frac{\sum_{i=1}^{n}\boldsymbol{\tau}(x_{i-1},x_i)}{n}\in\mathbb{R}^d$ is a minimal sufficient statistic. Through the same approach as in Section {\ref{sec::TSC}}, we partition the convex hull of the space of minimal sufficient statistics into cuboids of side length $\frac{s}{n}$ defined as in (\ref{cuboidEq}). We then characterize quantized type classes as in (\ref{typeClassDefEq}).

Let
\begin{equation}
\label{entropyRate}
H(p_{\theta})=\lim_{n\rightarrow\infty}\frac{1}{n}\mathbb{E}_{\theta}\left[\log{\frac{1}{p_{\theta}(X^n)}}\right]
\end{equation}
and
\begin{equation}
\label{varentropyRate}
\sigma^2(p_{\theta})=\lim_{n\rightarrow\infty}\frac{1}{n}\mathbb{V}_{\theta}\left[\log{\frac{1}{p_{\theta}(X^n)}}\right]
\end{equation}
be the entropy and the varentropy rate of the Markov process parametrized by $\theta$, respectively. The following theorem characterizes the fundamental limits of universal one-to-one compression of parametric Markov sources, as well as asserting that the TS code is optimal up to the third-order term.
\begin{Theorem}
\label{markovTheorem}
For any first-order, stationary, irreducible and aperiodic Markov exponential model class parametrized by $\Theta_{\mathcal{M}}$
\begin{equation*}
\inf_{\phi}\sup_{\theta\in\Theta_{\mathcal{M}}}\left[R_n(\epsilon,\phi,p_{\theta})-H(p_{\theta})-\frac{\sigma(p_{\theta})}{\sqrt{n}}Q^{-1}(\epsilon)\right]= \left(\frac{d}{2}-1\right)\frac{\log{n}}{n}+\mathcal{O}\left(\frac{1}{n}\right).\nonumber
\end{equation*}
where the infimum is achieved by the quantized type class implementation of the TS code.
\end{Theorem}
\begin{proof}
Let $Y_i=(X_{i-1},X_i)$ be a random vector defined by overlapping blocks of $\{X_n\}$. Since ${X_n}$ form a Markov chain, so does $\{Y_n\}$. The proof follows the same lines as those in the proof of the parametric $i.i.d.$ class $\mathcal{P}$, with $\boldsymbol{\tau}(Y_n)$ playing the role of $\boldsymbol{\tau}(X_n)$. The only deviations from the memoryless proof occur in lines (\ref{stoneEq}), (\ref{dEq}) and (\ref{converseBerryEsseen}). As a counterpart of the $i.i.d.$ ratio limit theorem of (\ref{stoneEq}) for a Markov sources, we may use Theorem 8 of \cite{korshunov}, which states that
\begin{equation*}
p_{\hat{\theta}_c(x^n)}\left\{n\boldsymbol{\tau}(Y^n)\in nG(\boldsymbol{\tau}_c))\right\}= \frac{s^d}{\left(2\pi n\right)^{\frac{d}{2}}|\boldsymbol{\Sigma}|^{\frac{1}{2}}}e^{-\frac{\left\langle\left(x-n\boldsymbol{\mu}\right)\boldsymbol{\Sigma}^{-1},x-n\boldsymbol{\mu}\right\rangle}{2n}}+o\left(n^{-\frac{d}{2}}\right)
\end{equation*}
where $\boldsymbol{\Sigma}$ and $\boldsymbol{\mu}$ are the covariance and mean of the stationary distribution of the Markov chain, respectively. Finally (\ref{dEq}) and (\ref{converseBerryEsseen}) can be derived from the Markov version of the normal approximation inequality stated below. The proof is the same as in Appendix \ref{app::maxLikeBerr}.

\begin{Lemma}[Asymptotic Normality of Information]\label{berryLemma}
Fix a positive constant $\alpha$. For a first-order, stationary, irreducible and aperiodic Markov source, there exists a finite positive constant $A'$ such that for all $n\geq 1$ and $z$ such that $|z|\leq \alpha$,
\begin{equation}
\left|\mathbb{P}_{\theta^*}\left\{\frac{-\log{p_{\hat{\theta}(X^n)}(X^n)}-nH}{\sqrt{n}\sigma}>z\right\}-Q(z)\right|\leq \frac{A'}{\sqrt{n}}
\end{equation}
where $H:=H(p_{\theta^*})$ and $\sigma^2:=\sigma^2(p_{\theta^*})$, are the entropy and varentropy rate of the true model, $p_{\theta^*}$, respectively.
\end{Lemma}

The rest of the proof is the same as the $i.i.d.$ case and we omit it due to similarity.
\end{proof}

\begin{Example}
For the class of all first-order stationary, irreducible and aperiodic Markov sources $d=|\mathcal{X}|\left(|\mathcal{X}|-1\right)$, and Theorem \ref{markovTheorem} reduces to the result in \cite{nemat}.
\end{Example}

\section{Type Size Code with Point Type Classes}
\label{sec::AltApprch}
In this section we analyze the performance of the point type class implementation of the TS code. For a sequence $x^n\in\mathcal{X}^n$, define the point type class containing $x^n$ as
\begin{equation}
\label{merhavTypeClass}
T_{x^n}=\left\{y^n\in\mathcal{X}^n: p_{\theta}(x^n)=p_{\theta}(y^n) \mbox{ for all } \theta\in\Theta\right\}
\end{equation}
the set of all $n$-length sequences $y^n\in\mathcal{X}^n$ equiprobable with $x^n$, simultaneously under all models in $\mathcal{P}$. Consequently, (\ref{pNdimEq}) enforces two sequences to be in the same type class if and only if their minimal sufficient statistics are equal. Hence, from a geometric perspective, point type classes correspond to zero sidelength $s=0$ in Figure \ref{fig::TClassPar}, i.e. type classes are points in the space of minimal sufficient statistics. We first review the derivation of the size of a point type class from \cite{merhav}. We then provide upper and lower bounds for the asymptotic rate of the TS code with point type class implementation, showing that the TS code performs strictly worse for $s=0$ in terms of third-order coding rate.

Let $\boldsymbol{\tau}(x)[j]$, $j=1,\cdots,d$, be the $j$-th component of the $d$-dimensional vector $\boldsymbol{\tau}(x)$. For any index $j=1,\cdots,d$, there exists a fixed real number $\beta[j][0]$ and $r_j$ pairwise incommensurable real numbers $\beta[j][t]$, $t=1,\cdots,r_j$, such that regardless of the observed sample $x\in\mathcal{X}$, $\boldsymbol{\tau}(x)[j]$ can be uniquely decomposed as \cite{merhav}
\begin{equation}
\label{decomEq}
\boldsymbol{\tau}(x)[j] = \beta[j][0]+\sum_{t=1}^{r_j}\beta[j][t]\tilde{L}(x)[j][t]
\end{equation}
where $\tilde{L}(x)[j][t]$, $t=1,\cdots,r_j$, are integers depending on the sample $x$ through $\boldsymbol{\tau}(x)[j]$. The decomposition (\ref{decomEq}) defines a unique one-to-one mapping between the real-valued $\boldsymbol{\tau}(x)[j]$ and $r_j$ integers $\tilde{L}(x)[j][t]$. Concatenating the corresponding unique integers $\tilde{L}(x)[\cdot][\cdot]$, each $d$-dimensional vector $\boldsymbol{\tau}(x)$ corresponds to a unique integer-valued vector $\tilde{\boldsymbol{L}}(x)\in\mathbb{Z}^{\sum_{j=1}^{d}r_j}$. For all $j=1\cdots d$, we may choose without loss of generality $ \beta[j][0]=\boldsymbol{\tau}(1)[j]$. With this choice we always have $\tilde{\boldsymbol{L}}(1)=(0,\cdots,0)^T$. Let $d'$, which is called the dimensionality of the type class in \cite{merhav}, be the rank of the matrix
${\mathbb{\tilde{L}}}=\begin{bmatrix}
         \tilde{\boldsymbol{L}}(2)-\tilde{\boldsymbol{L}}(1) &   \cdots & \tilde{\boldsymbol{L}}(|\mathcal{X}|)-\tilde{\boldsymbol{L}}(1)
        \end{bmatrix}$. Therefore, there are $d'$ linearly independent rows in $\tilde{\mathbb{L}}$. Let the indices of the linearly independent rows be $i_1,\cdots,i_{d'}$. For any $x\in\mathcal{X}$, define $d'$-dimensional vector $\boldsymbol{L}(x)$ as $\boldsymbol{L}(x)[j]=\tilde{\boldsymbol{L}}(x)[i_j]$ for $j=1\cdots d'$. Since the other rows are linear combination of the independent rows, we can denote this transformation as $\tilde{\mathbb{L}}=\boldsymbol{R}\mathbb{L}$, where $\boldsymbol{R}$ is a $\sum_{j=1}^{d}{r_j}\times d'$ matrix and $\mathbb{L}$ is a full-rank $d'\times(|\mathcal{X}|-1)$ dimensional matrix $\mathbb{L}=\begin{bmatrix}
         {\boldsymbol{L}}(2)-{\boldsymbol{L}}(1) &   \cdots & {\boldsymbol{L}}(|\mathcal{X}|)-{\boldsymbol{L}}(1)
        \end{bmatrix}$. Since $\tilde{\boldsymbol{L}}(1)=\boldsymbol{L}(1)=\boldsymbol{0}$, there is a one to one correspondence between $\boldsymbol{L}(x)$ and $\tilde{\boldsymbol{L}}(x)$ and consequently between $\boldsymbol{L}(x)$ and $\boldsymbol{\tau}(x)$.

Note that $d'\geq d$, and in many cases the inequality is strict. The main finding of this section is that $d'$ is the critical dimension for the behavior of the TS code under point type classes, rather than $d$. Since $d'$ may be larger than $d$, the performance of the TS code with point type classes may be strictly worse than that with quantized type classes.

Let $\boldsymbol{\mathbf{b}}$ be a $d\times 1$ column vector containing $\beta[j][0]$'s for $j=1,\cdots,d$ and $\boldsymbol{\mathbb{A}}$ is a $d\times \sum_{j=1}^{d}{r_j}$ block diagonal matrix containing $\beta[j][t]$'s in (\ref{decomEq}). For real-valued vector $\ell\in\mathbb{R}^{d'}$, let $\boldsymbol{\tau}(\ell)=\boldsymbol{\mathbf{b}}+\boldsymbol{\mathbb{A}}\boldsymbol{R}\ell$. For a constant $C>0$ to be defined later, define $f_0(\ell)$ as follows:
\begin{align}
f_0(\ell)&= -\frac{1}{n} \left(\left\langle \hat{\theta}(\boldsymbol{\tau}(\ell)),\boldsymbol{\tau}(\ell) \right\rangle - \psi\left(\hat{\theta}(\boldsymbol{\tau}(\ell))\right)\right) -\frac{d'}{2n}\log{2\pi n} +\frac{C}{n} \label{subnserf} \\
 &=-\frac{1}{n} \left(\left\langle \hat{\theta}(\boldsymbol{\mathbf{b}}+\boldsymbol{\mathbb{A}} \boldsymbol{R} \ell),\boldsymbol{\mathbf{b}}+\boldsymbol{\mathbb{A}}\boldsymbol{R}\ell \right\rangle - \psi\left(\hat{\theta}\left(\boldsymbol{\mathbf{b}}+\boldsymbol{\mathbb{A}}\boldsymbol{R}\ell\right)\right)\right) -\frac{d'}{2n}\log{2\pi n} +\frac{C}{n} \label{secondContReal}.
\end{align}

For a sequence $x^n$, define $\boldsymbol{L}(x^n)$ similar to (\ref{tauXnDef}) as
\begin{equation}
\label{LSumEq}
\boldsymbol{L}(x^n)=\frac{\sum_{i=1}^{n}\boldsymbol{L}(x_i)}{n}
\end{equation}
and let $\mathcal{L}=\left\{\boldsymbol{L}(x^n):x^n\in\mathcal{X}^n\right\}$ be the set of lattice points. Throughout, $\boldsymbol{L}\in\mathbb{Z}^{d'}$ denotes an integer-valued lattice point, while $\ell\in\mathbb{R}^{d'}$ denotes real-valued $d'$-dimensional vector.

The size of a point type class is derived in \cite{merhav}, which we reproduce it in Appendix \ref{app::pointProof} for completeness. Moreover, we show that the third-order term in their result is a constant to obtain the following lemma.
\begin{Lemma}
\label{PointTypeClasLemma}
For large enough $n$, the size of the point type class containing $x^n$ with $\boldsymbol{L}(x^n)=\boldsymbol{L}$, is bounded as
\begin{equation}
\label{typefNotBnd}
nf_0(\boldsymbol{L})-2C\leq \log{|T_{x^n}|} \leq  nf_0(\boldsymbol{L})
\end{equation}
where $C$ is the constant in (\ref{subnserf}, \ref{secondContReal}).
\end{Lemma}
\begin{proof}
See Appendix \ref{app::pointProof}.
\end{proof}
The following is our main theorem for this section, characterizing the exact performance of the TS code with point type classes up to third-order.
\begin{Theorem}
\label{latticeTSCThm}
Let $\phi_0$ be the point type class implementation of the TS code. The $\epsilon$-coding rate of $\phi_0$, for all $\theta\in\Theta$ is given by
\begin{equation}
R_n(\epsilon,\phi_0,{p_{\theta}})=H(p_{\theta})+\frac{\sigma(p_{\theta})}{\sqrt{n}}Q^{-1}(\epsilon) +\left(\frac{d'}{2}-1\right)\frac{\log{n}}{n}+\mathcal{O}\left(\frac{1}{n}\right). \label{mainEq}
\end{equation}
\end{Theorem}
\begin{proof}
The achievability proof is similar to Section \ref{sec::proofMain}, hence we only highlight the differences. Again for simplicity, we denote $H=H(p_{\theta^*})$ and $\sigma=\sigma(p_{\theta^*})$ as the entropy and the varentropy of the underlying model $p_{\theta^*}$, respectively. Let
\begin{equation}
\label{gammaPrimeEq}
\gamma'= H+\frac{\sigma}{\sqrt{n}}Q^{-1}\left(\epsilon-\frac{A}{\sqrt{n}}\right)-\frac{d'}{2n}\log\left(2\pi n\right) +\frac{C}{n}.
\end{equation}
We now show that for this choice of $\gamma'$, $p_{\gamma'}\leq \epsilon$, where $p_{\gamma'}$ is defined as in (\ref{pGammaFirst}). We have
\begin{align}
p_{\gamma'}&=\mathbb{P}_{\theta^*}\left[\log{|T_{X^n}|>n\gamma'}\right] \nonumber \\
        &=\mathbb{P}_{\theta^*}\left[\frac{-\log{p_{\hat{\theta}}(x^n)}-nH}{\sigma\sqrt{n}}>Q^{-1}\left(\epsilon-\frac{A}{\sqrt{n}}\right)\right] \label{typeLemmas} \\
				&\leq Q\left(Q^{-1}\left(\epsilon-\frac{A}{\sqrt{n}}\right)\right)+\frac{A}{\sqrt{n}} \label{berApp} \\
				&= \epsilon \nonumber
\end{align}
where (\ref{typeLemmas}) follows from (\ref{typefNotBnd}, \ref{subnserf}, \ref{gammaPrimeEq}) by noticing that
\begin{equation*}
f_0(\boldsymbol{L})=-\frac{1}{n}\log{p_{\hat{\theta}(x^n)}(x^n)}-\frac{d'}{2n}\log{(2\pi n)}+\frac{C}{n}
\end{equation*}
for any $x^n$ with $\boldsymbol{L}(x^n)=\boldsymbol{L}$, and (\ref{berApp}) is an application of Lemma \ref{maxLikeBerr}.

Recall that there is a one-to-one correspondence between $T_{x^n}$ and $\boldsymbol{L}(x^n)$, hence we can denote $T_{x^n}$ as $T_{\boldsymbol{L}(x^n)}$. Furthermore, once $x^n$ is understood from the context, we simplify $T_{\boldsymbol{L}(x^n)}$ and rewrite it as $T_{\boldsymbol{L}}$. We can then reformulate the equation for $M(\epsilon)$ in (\ref{mEq}) for point type classes. We can achieve this, simply by replacing $\boldsymbol{\tau}_c(X^n)$ with $\boldsymbol{L}(x^n)$ as the representative of the type class.

We then bound $M(\epsilon)$ in (\ref{mEq}) with the choice of $\gamma'$ in (\ref{gammaPrimeEq}). Through the same approach as in Subsection \ref{subsec::Achiev}, one can show that
\begin{equation}
M(\epsilon)\leq  \sum_{i=0}^{\infty}\left|\left\{\boldsymbol{L}\in\mathcal{L}:f_0(\boldsymbol{L})\in\mathcal{A}'_i\right\}\right| \cdot  2^{\left\{n\gamma'+2C-ni\Delta\right\}} \label{newmEpsilon}
\end{equation}
where $\mathcal{A}'_i=\left(\gamma'+\frac{2C}{n}-(i+1)\Delta,\gamma'+\frac{2C}{n}-i\Delta\right]$ and $C$ is the constant in (\ref{typefNotBnd}). We now evaluate $\left|\left\{\boldsymbol{L}\in\mathcal{L}:f_0(\boldsymbol{L})\in\mathcal{A}'_i\right\}\right|$. Define a 2-norm ball of radius $r$ around a point $\ell_0\in\mathbb{R}^{d'}$ as
\begin{equation}
\label{ballsEq}
B_{r}(\ell_0)=\left\{\ell\in\mathbb{R}^{d'}:\|\ell-\ell_0\|<r\right\}.
\end{equation}
In the sequel we use $\boldsymbol{L}$ as the lattice points in $\mathcal{L}$, while we reserve the notation $\ell$ for points in the convex hull of $\mathcal{L}$ which we denote by $\mathfrak{L}=\text{conv}(\mathcal{L})$. Observe that for any two different points $\boldsymbol{L}_1,\boldsymbol{L}_2\in\mathcal{L}$, $\|\boldsymbol{L}_1-\boldsymbol{L}_2\|\geq \frac{1}{n}$, and  therefore, $B_{\frac{1}{2n}}(\boldsymbol{L}_1)$ and $B_{\frac{1}{2n}}(\boldsymbol{L}_2)$ are disjoint. Since the convex hull $\mathfrak{L}$ is a $d'$-dimensional space, there exists a constant $C>0$ (its precise value is $\frac{\pi^{\frac{d'}{2}}}{2^{d'}\Gamma(\frac{d'}{2}+1)}$ \cite{ren}) such that
\begin{equation}
\label{volumeEquation}
\text{Vol}\left(B_{\frac{1}{2n}}(\boldsymbol{L})\right)=\frac{C}{n^{d'}}.
\end{equation}
Therefore
\begin{align}
|\left\{\boldsymbol{L}\in\mathcal{L}:f_0(\boldsymbol{L})\in\mathcal{A}_i'\right\}|&=\sum_{\substack{\boldsymbol{L}\in\mathcal{L}\\f_0(\boldsymbol{L})\in\mathcal{A}_i'}}\frac{n^{d'}}{C}\text{Vol}\left(B_{\frac{1}{2n}}(\boldsymbol{L})\right) \nonumber \\
&=\frac{n^{d'}}{C}\text{Vol}\left(\bigcup_{\substack{\boldsymbol{L}\in\mathcal{L}\\f_0(\boldsymbol{L})\in\mathcal{A}'_i}}B_{\frac{1}{2n}}(\boldsymbol{L})\right) \label{ballDisjointness}\\
             &\leq \frac{n^{d'}}{C}\text{Vol}\left(\bigcup_{\substack{\ell\in\mathfrak{L}\\f_0(\ell)\in\mathcal{A}'_i}}B_{\frac{1}{2n}}(\boldsymbol{L})\right).
             \nonumber
\end{align}
where (\ref{ballDisjointness}) follows from disjointness of the balls. Proceeding as in Subsection \ref{subsec::Achiev}, it is straightforward to show that for a constant $C>0$
\begin{equation}
\label{numberOfLs}
|\left\{\boldsymbol{L}\in\mathcal{L}:f_0(\boldsymbol{L})\in\mathcal{A}'_i\right\}| \leq Cn^{d'-1}.
\end{equation}
The rest of the proof is similar to the Subsection \ref{subsec::Achiev}, which we omit due to similarity.

We now provide a converse for the performance of the Type Size code with point type classes. We can rewrite the corresponding finite blocklength result (\ref{mEq}) for point type classes as
\begin{equation}
\label{secondMEpsilon}
M(\epsilon)=\inf_{\gamma':p_{\gamma'}\leq \epsilon}v(\gamma'),
\end{equation}
where $p_{\gamma'}$ is defined as in (\ref{pGammaFirst}) and
\begin{equation}
v(\gamma')=\sum_{\substack{\boldsymbol{L}\in\mathcal{L}: \\ \frac{1}{n}\log{|T_{\boldsymbol{L}}|}\leq \gamma' }}{{|T_{\boldsymbol{L}}|}}.
\end{equation}
Notice that $v(\gamma')$ is non-decreasing function of $\gamma'$, while $p_{\gamma'}$ is non-increasing function of $\gamma'$. Therefore, if for some $\gamma'_0$, $p_{\gamma'_0}>\epsilon$, then one can conclude that
\begin{equation}
M(\epsilon)\geq v(\gamma_0'). \label{mEpsVRel}
\end{equation}
We then show that $p_{\gamma_0'}>\epsilon$ for the following choice of $\gamma_0'$
\begin{equation}
\label{gammaPrimeZeroEq}
\gamma'_0=H+\frac{\sigma}{\sqrt{n}}Q^{-1}\left(\epsilon+\frac{A+1}{\sqrt{n}}\right)-\frac{d'}{2n}\log{(2\pi n)}-\frac{C}{n}
\end{equation}
where $A$ is the constant in Lemma \ref{maxLikeBerr} and $C$ is the constant in (\ref{typefNotBnd}).
Indeed
\begin{align}
p_{\gamma'_0} &\geq\mathbb{P}_{\theta^*}\left[-\frac{1}{n}\log{p_{\hat{\theta}(X^n)}(X^n)}-\frac{d'}{2n}\log{(2\pi n)}-\frac{C}{n}>\gamma'_0\right] \label{firstLowLam} \\
&=\mathbb{P}_{\theta^*}\left[\frac{-\log{p_{\hat{\theta}(X^n)}(X^n)}-nH}{\sigma\sqrt{n}}>Q^{-1}\left(\epsilon+\frac{A+1}{\sqrt{n}}\right)\right] \label{secondLowLam} \\
&>\epsilon \label{thirdLowLam}
\end{align}
where (\ref{firstLowLam}) is from the type class size bound (\ref{typefNotBnd}) and the definition of $p_{\gamma'_0}$ in (\ref{pGammaFirst}), (\ref{secondLowLam}) is from the choice of $\gamma'_0$ in (\ref{gammaPrimeZeroEq}), and (\ref{thirdLowLam}) is a consequence of Lemma \ref{maxLikeBerr}. Continuing from (\ref{mEpsVRel}), we may write
\begin{align}
M(\epsilon)&\geq \sum_{\substack{\boldsymbol{L}\in\mathcal{L} \\ \frac{1}{n}\log{|T_{\boldsymbol{L}}|}\leq \gamma'_0}}{|T_{\boldsymbol{L}}|} \nonumber \\
           &\geq \sum_{\substack{\boldsymbol{L}\in\mathcal{L} \\ f_0(\boldsymbol{L})\leq \gamma'_0}}{2^{nf_0(\boldsymbol{L})-2C}} \label{mfarSec}
\end{align}
where (\ref{mfarSec}) exploits the bounds for the type class size (\ref{typefNotBnd}). For $\Delta=\frac{1}{n}$, (\ref{mfarSec}) can simply be lower bounded as follows by restricting the summation to $\boldsymbol{L}$ in $\mathcal{A}_0$, where $\mathcal{A}_0=\{\boldsymbol{L}\in\mathcal{L}: \gamma'_0-\Delta< f_0(\boldsymbol{L})\le \gamma'_0\}$
\begin{equation}
M(\epsilon)\geq \left|\mathcal{A}_0\right| \cdot  2^{n\gamma'_0-n\Delta-2C}. \label{newmEpsilon}
\end{equation}

We now provide a lower bound on $|\mathcal{A}_0|$. Let $\tilde{\mathcal{A}}_0=\{\ell\in\mathfrak{L}: \gamma'_0-\Delta< f_0(\ell)\le \gamma'_0\}$.
\begin{Lemma}
\label{disConLem}
There exists a constant $C$ such that
\begin{equation}
\label{disConEq}
\frac{\text{\emph{Vol}}\left(\bigcup_{\ell\in\tilde{\mathcal{A}}_0}B_{\frac{1}{2n}}(\ell)\right)}{\text{\emph{Vol}}\left(\bigcup_{\boldsymbol{L}\in\mathcal{A}_0}B_{\frac{1}{2n}}\left(\boldsymbol{L}\right)\right)}\leq C.
\end{equation}
\end{Lemma}
\begin{proof}
See Appendix \ref{app::disConLem}.
\end{proof}

We then have
\begin{align}
|\mathcal{A}_0|
            &=\sum_{\boldsymbol{L}\in\mathcal{A}_0}\frac{\text{Vol}\left(B_{\frac{1}{2n}}(\boldsymbol{L})\right)}{\text{Vol}\left(B_{\frac{1}{2n}}(\boldsymbol{L})\right)} \nonumber \\
            &= C n^{d'}\sum_{\boldsymbol{L}\in\mathcal{A}_0}\text{Vol}\left(B_{\frac{1}{2n}}\left(\boldsymbol{L}\right)\right)  \label{e1}\\
            &= C n^{d'}\text{Vol}\left(\bigcup_{\boldsymbol{L}\in\mathcal{A}_0}B_{\frac{1}{2n}}\left(\boldsymbol{L}\right)\right)  \label{e2}\\
            &\geq  n^{d'}\frac{\text{Vol}\left(\bigcup_{\ell\in\tilde{\mathcal{A}}_0}B_{\frac{1}{2n}}(\ell)\right)}{C} \label{e3}
\end{align}
where (\ref{e1}) follows from (\ref{volumeEquation}) (recall that the letter $C$ may denote different constants), (\ref{e2}) is due to disjointness of the balls, and (\ref{e3}) is a consequence of Lemma \ref{disConLem}. Define,
\begin{equation}
\rho_0(\lambda)=\text{Vol}\{\ell\in \mathfrak{L}:f_0(\ell)\leq \lambda\}.
\end{equation}
We need the following technical lemma, which we prove in Appendix \ref{sec::rhoZeroAppx}.
\begin{Lemma}
\label{rho0LowerDrive}
There exists a positive constant $K_4$, such that for all $\gamma'_0-\Delta\leq\lambda \leq \gamma'_0$ we have
\begin{equation*}
\left|\frac{d}{d\lambda}\rho_0(\lambda)\right|\geq K_4.
\end{equation*}
\end{Lemma}
Recalling the definition of $\tilde{\mathcal{A}}_0$, we may continue from (\ref{e3}) and write
\begin{align}
|\mathcal{A}_0| &\geq \frac{n^{d'}}{C}\text{Vol}\left(\cup_{\ell:\gamma_0'-\Delta<f_0(\ell)\leq\gamma_0'}B_{\frac{1}{2n}}(\ell)\right) \nonumber \\
                                                                 &\geq \frac{n^{d'}}{C}\text{Vol}\left(\{\ell:f_0(\ell)\in\left(\gamma_0'-\Delta,\gamma_0'\right]\}\right)  \label{eqqq1}\\
                                                                 &=\frac{n^{d'}}{C}\left(\rho_0(\gamma_0')-\rho_0(\gamma_0'-\Delta)\right) \label{eqqq2} \\
                                                                 &\geq \frac{n^{d'}}{C}K_4\Delta   \label{eqqq3}
\end{align}
where (\ref{eqqq1}) is by lower bounding the volume of the ball-covering of $\tilde{\mathcal{A}}_0$ by the volume of $\tilde{\mathcal{A}}_0$ itself, (\ref{eqqq2}) is from the definition of $\rho_0$ and (\ref{eqqq3}) is from Lemma \ref{rho0LowerDrive}.

Continuing from (\ref{newmEpsilon}), we have
\begin{align}
M(\epsilon) &\geq \frac{n^{d'}}{C}K_4\Delta \cdot  2^{n\gamma'_0-n\Delta-2C} \label{mepOne} \\
             &=Cn^{d'-1}2^{n\gamma'_0-2C-1} \label{mepthree}
\end{align}
where (\ref{mepOne}) is from (\ref{eqqq3}), and (\ref{mepthree}) is from $\Delta=\frac{1}{n}$. Replacing $\gamma'_0$ by (\ref{gammaPrimeZeroEq}), we obtain
\begin{equation*}
\log{M(\epsilon)}\geq nH+\sigma\sqrt{n}Q^{-1}(\epsilon)+\left(\frac{d'}{2}-1\right)\log{n}+\mathcal{O}(1).
\end{equation*}
\end{proof}

\section{Conclusion and Future Work}
\label{sec::conclusion}
We derived the fundamental limits for universal one-to-one coding of the $d$-dimensional memoryless as well as Markov exponential families of distributions. We proposed the quantized Type Size code, where type classes are associated with cuboids in the grid partitioning the space of minimal sufficient statistics. We showed that the quantized Type Size code achieves the optimal third-order term $\left(\frac{d}{2}-1\right)\log{n}$. Next, the naive point type class approach is considered, where two sequences are in the same type class if and only if they have the same probability under any distribution in the exponential family. In the point type class scenario, each point (rather than a cuboid) in the set of minimal sufficient statistics defines a type class. The third-order term of the point type class approach is shown to be exactly $(\frac{d'}{2}-1)\log{n}$, where $d'$ is the dimension of the lattice vector representation of the the sufficient statistic. Since $d'$ is in general larger than $d$, our findings reveal that the model class dimension $d$ --- rather than the lattice dimension $d'$ --- is the relevant dimension for optimal performance. This is a more intuitive result, because it is much easier to understand the role of $d$ as opposed to $d'$. Moreover, $d$ is a more robust parameter compared to $d'$; changing the model parameters infinitesimally (i.e. from rational to irrational) can change $d'$, but not $d$.

For a more general parametric family without any information on the minimal sufficient statistics, one may partition the parameter space into cuboids and define two sequences to be in the same type class if and only if their maximum likelihood estimates belong to the same cuboid. One interesting future direction of this work is analyzing performance of such approach. As this work does not consider computational complexity of implementing the compression algorithms, an alternative future direction is to consider the blocklength-storage-complexity tradeoff. Finally, the lossy version of this research is also an interesting possible future direction.

\begin{appendices}

\section{Proof of Lemma \ref{maxLikeBerr}: Asymptotic Normality of Information}\label{app::maxLikeBerr}
Define
\begin{equation}
\label{eeEqau}
e(\boldsymbol{\tau}) = -\max_{\theta} \Big(\langle\theta,\boldsymbol{\tau}\rangle-\psi(\theta)+\langle\theta,\nabla\psi(\theta^*)\rangle\Big).
\end{equation}
Fuethermore, denote $\boldsymbol{U}_i(x^n)=\boldsymbol{\tau}(x_i)-\boldsymbol{\mu}$ for $i=1,\cdots ,n$, where $\boldsymbol{\mu}=\mathbb{E}_{\theta^*}\left[\boldsymbol{\tau}(X)\right]$. Therefore, $\boldsymbol{U}_i(X^n)$'s are zero-mean with finite covariance. First, observe that
\begin{align}
\frac{1}{n}\log{p_{\hat{\theta}}(x^n)}&= -\max_{\theta}\langle \theta,\boldsymbol{\tau}(x^n)-\boldsymbol{\mu} \rangle -\psi(\theta) +\langle \theta,\boldsymbol{\mu}\rangle \label{diffSum}\\
                           &=e\left(\frac{1}{n}\sum_{i=1}^{n}{\boldsymbol{U}_i(x^n)}\right) \label{sumUI}
\end{align}
where (\ref{diffSum}) is from (\ref{thetaHatEquation}), and since $\boldsymbol{\mu}=\nabla\psi(\theta^*)$ \cite{jordan}, (\ref{sumUI}) follows from (\ref{eeEqau},\ref{tauXnDef}).
We then show that $e(\mathbf{0})=H$. Equating the derivative with respect to $\theta$ of the expression inside the parenthesis with zero, we find that $\theta^*$ is the maximizing parameter in (\ref{eeEqau}). Therefore
\begin{align}
e(\mathbf{0}) &= -\Big(-\psi(\theta^*)+\langle \theta^*,\nabla\psi(\theta^*) \rangle\Big) \label{thetaStarIsMax}\\
     &= -\Big(-\psi(\theta^*)+\langle \theta^*,\mathbb{E}_{\theta^*}\left(\boldsymbol{\tau}(X)\right) \rangle \Big)\label{ExoProp}\\
     &=-\mathbb{E}_{\theta^*}\Big(-\psi(\theta^*)+\langle \theta^*,\left(\boldsymbol{\tau}(X)\right)\rangle\Big) \nonumber\\
     &=-\mathbb{E}_{\theta^*}\Big(\log{p_{\theta^*}(X)}\Big) \label{useFirst} \\
     &=H\nonumber
\end{align}
where (\ref{thetaStarIsMax}) is from (\ref{eeEqau}), (\ref{ExoProp}) is an exponential family property \cite{jordan}, and (\ref{useFirst}) is from (\ref{pThetaEq}). Application of the Proposition 1 in \cite{nemat} completes the proof.

\section{Proof of Lemma \ref{apprxLemm}: Maximum Likelihood Approximation}
\label{app::apprxLemm}
We show that $\log p_{\hat{\theta}(x^n)}(x^n)$ is constant away from $\log p_{\hat{\theta}_c(x^n)}(x^n)$. Recall that
\begin{equation*}
\log{p_{\hat{\theta}(x^n)}(x^n)}=n\max_{\theta}\Big[\langle\theta,\boldsymbol{\tau}(x^n)\rangle-\psi(\theta)\Big].
\end{equation*}
For ease of notation, when it is clear from the context, we denote $\boldsymbol{\tau}_c(x^n)$ as $\boldsymbol{\tau}_c$, and similarly we remove the argument in $\hat{\theta}_c(x^n)$ and simply denote it as $\hat{\theta}_c$.
Since $\boldsymbol{\tau}(x^n)$ is in a cuboid of side length $\frac{s}{n}$ with center $\boldsymbol{\tau}_c$, we have $\|\boldsymbol{\tau}(x^n)-\boldsymbol{\tau}_c\|\leq \frac{s\sqrt{d}}{2n}$. We hence have
\begin{align}
\left|\left\langle\hat{\theta}_c,\boldsymbol{\tau}(x^n)\right\rangle - \left\langle\hat{\theta}_c,\boldsymbol{\tau}_c\right\rangle\right|&=\left|\left\langle \hat{\theta}_c, \boldsymbol{\tau}(x^n)-\boldsymbol{\tau}_c\right\rangle\right| \nonumber \\
&\leq \|\hat{\theta}_c\| \|\boldsymbol{\tau}(x^n)-\boldsymbol{\tau}_c\|  \nonumber \\
&\leq \wp\frac{s\sqrt{d}}{2n}  = \frac{\kappa s}{n} \label{tauProdVSTauCProd}
\end{align}
where (\ref{tauProdVSTauCProd}) exploits the fact that $\|\theta\|\leq\wp$, for all $\theta\in\Theta$, including $\hat{\theta}_c$. Therefore
\begin{align}
\log{p_{\hat{\theta}_c(x^n)}(x^n)}&=n\left[\left\langle\hat{\theta}_c,\boldsymbol{\tau}(x^n)\right\rangle-\psi\left(\hat{\theta}_c\right)\right] \nonumber \\
                                  &\geq n\left[\left\langle\hat{\theta}_c,\boldsymbol{\tau}_c(x^n)\right\rangle-\frac{\kappa s}{n}-\psi\left(\hat{\theta}_c\right)\right] \label{aaaEq}\\
														 &= n\max_{\theta}\Big[\left\langle\theta,\boldsymbol{\tau}_c(x^n)\right\rangle-\frac{\kappa s}{n}-\psi(\theta)\Big] \label{ThetaCResCor}
\end{align}
where (\ref{aaaEq}) follows from (\ref{tauProdVSTauCProd}) and (\ref{ThetaCResCor}) is from the definition of $\hat{\theta}_c$. Using the fact that for any two functions $g_1(\theta),g_2(\theta)$
\begin{equation}
\max_{\theta}g_1(\theta)-\max_{\theta}g_2(\theta)\leq \max_{\theta}\Big((g_1-g_2)(\theta)\Big) \label{maxDiff}
\end{equation}
we obtain
\begin{align}
\log{p_{\hat{\theta}(x^n)}(x^n)}-\log{p_{\hat{\theta}_c}(x^n)} &\leq n\max_{\theta}\Big[\left\langle\theta,\boldsymbol{\tau}(x^n)\right\rangle-\psi(\theta)\Big] -n\max_{\theta}\left[\left\langle\theta,\boldsymbol{\tau}_c(x^n)\right\rangle-\frac{\kappa s}{n}-\psi(\theta)\right]  \label{udePriorEq}\\
&\leq n\max_{\theta}\left[\left\langle\theta,\boldsymbol{\tau}(x^n)-\boldsymbol{\tau}_c\right\rangle+\frac{\kappa s}{n}\right] \label{fstLngEq}
\end{align}
where (\ref{udePriorEq}) exploits (\ref{ThetaCResCor}), and (\ref{fstLngEq}) is from (\ref{maxDiff}). Similar to (\ref{tauProdVSTauCProd}), one can show that $\langle\theta,\boldsymbol{\tau}(x^n)-\boldsymbol{\tau}_c\rangle\leq \frac{\kappa s}{n}$. Lemma then follows.

\section{Proof of Lemma \ref{fIsLipschitz}: Lipschitzness of $f(\cdot)$}
\label{app:fLipschitz}
Let
\begin{equation}
\label{lEq}
l(\boldsymbol{\tau})=\max_\theta \left(\left\langle\theta,\boldsymbol{\tau}\right\rangle-\psi(\theta)\right).
\end{equation}
Noticing that $\|\nabla f(\boldsymbol{\tau})\|=\|\nabla l(\boldsymbol{\tau})\|$, in order to show the Lipschitzness of $f(\boldsymbol{\tau})$ in (\ref{fEqUpp2}), it suffices to show that $l(\boldsymbol{\tau})$ is a Lipschitz function of $\boldsymbol{\tau}$. We first show that $\|\nabla l(\boldsymbol{\tau})\|=\|\hat{\theta}(\boldsymbol{\tau})\|$. Due to (\ref{thetaHatEquation})
\begin{equation*}
l(\boldsymbol{\tau})=\langle\hat{\theta}(\boldsymbol{\tau}),\boldsymbol{\tau}\rangle - \psi\left(\hat{\theta}(\boldsymbol{\tau})\right).
\end{equation*}
Hence, taking gradient with respect to $\boldsymbol{\tau}$
\begin{align}
\nabla l(\boldsymbol{\tau}) &=\left(\left(\nabla\hat{\theta}(\boldsymbol{\tau})\right) \boldsymbol{\tau}+\hat{\theta}(\boldsymbol{\tau})\right)-\nabla\hat{\theta}(\boldsymbol{\tau})\nabla_{\hat{\theta}}\psi\left(\hat{\theta}(\boldsymbol{\tau})\right) \nonumber \\
											&= \left(\left(\nabla\hat{\theta}(\boldsymbol{\tau})\right) \boldsymbol{\tau}+\hat{\theta}(\boldsymbol{\tau})\right)-\nabla\hat{\theta}(\boldsymbol{\tau})\mathbb{E}_{\hat{\theta}(\boldsymbol{\tau})}(\boldsymbol{\tau}(X)) \label{ssstackrela} \\
											&=\hat{\theta}(\boldsymbol{\tau}) \label{ssstackRellb}
\end{align}
where $(\ref{ssstackrela})$ follows from $\nabla_{\hat{\theta}}\psi\left(\hat{\theta}(\boldsymbol{\tau})\right)=\mathbb{E}_{\hat{\theta}(\boldsymbol{\tau})}\left(\boldsymbol{\tau}(X)\right)$ \cite{jordan}, and $(\ref{ssstackRellb})$ follows from $\mathbb{E}_{\hat{\theta}(\boldsymbol{\tau})}(\boldsymbol{\tau}(X))=\boldsymbol{\tau}$ (see the proof of Lemma \ref{TypeClSizeLem}). Lemma follows by recalling that $\|\hat{\theta}(\boldsymbol{\tau})\|\leq \wp$.

\section{Proof of Lemma \ref{rhoIsLipschitz}: Lipschitzness of $\rho(\cdot)$}
\label{app:hLipschitz}
Let $\mathcal{K}=\{\boldsymbol{\tau}\in\mathcal{T}:f(\boldsymbol{\tau})\leq \lambda\}$ and $\mathcal{K}^c=\mathcal{T}\backslash\mathcal{K}$.
We first show that $\mathcal{K}^c$ is a convex body. A sub-level set of $f(\cdot)$ is a sub-level set of $-l(\cdot)$ defined as in (\ref{lEq}). Therefore, it is enough to show that sub-level sets of $l(\cdot)$ (i.e. $\mathcal{K}^c$) are convex. Maximum of linear functions of $\boldsymbol{\tau}$ is a convex function, therefore $l(\cdot)$ defined in (\ref{lEq}) is a convex function of $\boldsymbol{\tau}$. Since the sub-level sets of a convex function are convex, $\mathcal{K}^c$ is a convex body.

In order to show that $\rho(\lambda)$ $\left(=\text{Vol}\left(\mathcal{K}\right)\right)$ is Lipschitz, we provide an upper bound for the absolute value of its derivative $|\frac{d}{d\lambda}\rho(\lambda)|$. Let us denote the surface area of a convex body $\mathcal{K}^c$ as \cite[Section 3.3]{hug}
\begin{equation}
\label{surfaceDef}
S(\mathcal{K}^c)=\lim_{\epsilon\rightarrow0}\frac{V^{(d)}\Big(\mathcal{K}^c+B(\epsilon)\Big)-V^{(d)}(\mathcal{K}^c)}{\epsilon}
\end{equation}
where $V^{(d)}(\cdot)$ is the $d$-dimensional volume, $B(\epsilon)$ is the $d$-dimensional unit ball and addition in $\mathcal{K}^c+B(\epsilon)$ is the Minkowski's sum \cite{hug}. Let us denote $\mathcal{K}^c_{\epsilon}=\{\boldsymbol{\tau}\in\mathcal{T}:f(\boldsymbol{\tau})> \lambda-\epsilon\}$. We have
\begin{align}
\frac{d}{d\lambda}\rho(\lambda) &= \lim_{\epsilon\rightarrow0}\frac{\rho(\lambda)-\rho(\lambda-\epsilon)}{\epsilon} \nonumber \\
               &=\lim_{\epsilon\rightarrow0}\frac{\left(\text{Vol}(\mathcal{T})-\rho(\lambda-\epsilon)\right)-\left(\text{Vol}(\mathcal{T})-\rho(\lambda)\right)}{\epsilon} \nonumber\\
               &= \lim_{\epsilon\rightarrow 0} \frac{\text{Vol}(\mathcal{K}^c_{\epsilon})-\text{Vol}(\mathcal{K}^c)}{\epsilon}. \label{hder}
\end{align}
Let us assume $\epsilon\rightarrow 0^+$; the case where $\epsilon\rightarrow 0^-$ is handled similarly. Let $\boldsymbol{\tau}_1\in\mathcal{K}^c_{\epsilon}$. From the Taylor series expansion of $f(\boldsymbol{\tau}_2)$ in the vicinity of $\boldsymbol{\tau}_1$ with distance at most $\|\boldsymbol{\tau}_2-\boldsymbol{\tau}_1\|\leq \sqrt{\epsilon}$, we obtain
\begin{equation}
f(\boldsymbol{\tau}_2)=f(\boldsymbol{\tau}_1)+ \langle \nabla f(\boldsymbol{\tau}_1),\boldsymbol{\tau}_2-\boldsymbol{\tau}_1\rangle+\Delta
\end{equation}
where $|\Delta|\leq C_f\|\boldsymbol{\tau}_1-\boldsymbol{\tau}_2\|^2$, for a constant $C_f$ independent of $n$. Let
\begin{equation}
\boldsymbol{\tau}_2=\boldsymbol{\tau}_1+\epsilon\frac{(1+C_f) \nabla f(\boldsymbol{\tau}_1)}{\|\nabla f(\boldsymbol{\tau}_1)\|^2}.
\end{equation}
With this choice of $\boldsymbol{\tau}_2$, we obtain

\begin{align}
f(\boldsymbol{\tau}_2)&=f(\boldsymbol{\tau}_1)+\epsilon (1+C_f)+\Delta \nonumber \\
					  &\geq f(\boldsymbol{\tau}_1)+\epsilon +\epsilon C_f- C_f\|\boldsymbol{\tau}_1-\boldsymbol{\tau}_2\|^2 \nonumber \\
                      &\geq f(\boldsymbol{\tau}_1)+\epsilon \label{epsilonDist} \\
					  &> \lambda-\epsilon+\epsilon \label{containRel} \\
                      &=\lambda \nonumber										
\end{align}
where (\ref{epsilonDist}) follows form $\|\boldsymbol{\tau}_2-\boldsymbol{\tau}_1\|\leq \sqrt{\epsilon}$, and (\ref{containRel}) is a consequence of $\boldsymbol{\tau}_1\in\mathcal{K}^c_{\epsilon}$. Hence $\boldsymbol{\tau}_2\in\mathcal{K}^c$. Since $\boldsymbol{\tau}_1\in\mathcal{K}^c_{\epsilon}$ was arbitrary, we have $\mathcal{K}^c_{\epsilon}\subset \mathcal{K}^c+B\left(\frac{\epsilon (1+C_f)}{\|\nabla f(\boldsymbol{\tau}_1)\|}\right)$. Therefore, one can upper bound (\ref{hder}) in terms of the surface area (\ref{surfaceDef}) as follows:
\begin{equation}
\label{hPrimeBound1}
\left|\frac{d}{d\lambda}\rho(\lambda)\right|\leq \frac{(1+C_f) S(\mathcal{K}^c)}{\|\nabla f(\boldsymbol{\tau}_1)\|} \:\:\:\mbox{ for all $\boldsymbol{\tau}_1\in\mathcal{K}_{\epsilon}$}.
\end{equation}
Since $\mathcal{K}^c,\mathcal{T}$ are convex bodies and $\mathcal{K}^c\subset \mathcal{T}$, consequently $S(\mathcal{K}^c)\leq S(\mathcal{T})$ \cite[Theorem 3.2.2]{hug}. Since $\mathcal{X}$ is finite, therefore $\mathcal{T}$ is a bounded set, which yields $S(\mathcal{K}^c)\leq S(\mathcal{T})<\infty$.

From the proof of Lemma \ref{fIsLipschitz} in Appendix \ref{app:fLipschitz}, we have $\|\nabla f(\boldsymbol{\tau}_1)\|=\|\hat{\theta}(\boldsymbol{\tau}_1)\|$. That translates (\ref{hPrimeBound1}) into
\begin{equation}
\label{hPrimeBound}
\left|\frac{d}{d\lambda}\rho(\lambda)\right|\leq \frac{(1+C_f) S(\mathcal{K}^c)}{\|\hat{\theta}(\boldsymbol{\tau}_1)\|} \:\:\:\mbox{ for all $\boldsymbol{\tau}_1\in\mathcal{K}_{\epsilon}$}.
\end{equation}
We finally show that $\|\hat{\theta}(\boldsymbol{\tau}_1)\|$ is bounded away from zero. Let $\boldsymbol{\tau}_u$ be such that $\hat{\theta}(\boldsymbol{\tau}_u)=(0,\cdots,0)$ (subscript $u$ stands for the uniform distribution.). Since $\omega=\frac{\log{|\mathcal{X}|}-H}{5}>0$ and $f(\boldsymbol{\tau}) = -\frac{1}{n}\log{p_{\hat{\theta}(\boldsymbol{\tau})}(x^n)}-\Theta\left(\frac{\log{n}}{n}\right)$, we have that
\begin{equation}
\label{fTauUEq}
f(\boldsymbol{\tau}_u)\geq \log{|\mathcal{X}|}-\omega, \mbox{ for large enough $n$}.
\end{equation}
From boundedness of $\mathcal{T}$, we have
\begin{equation*}
T_{\max}:=\max\left\{\|\boldsymbol{\tau}\|:\boldsymbol{\tau}\in \mathcal{T}\right\}<\infty.
\end{equation*}
Therefore $\left\|\nabla\psi\left(\hat{\theta}({\boldsymbol{\tau}})\right)\right\| = \left\|\mathbb{E}_{\hat{\theta}}(\boldsymbol{\tau}(X))\right\|\leq T_{\max}$ is bounded. Hence $\psi\left(\hat{\theta}({\boldsymbol{\tau}})\right)$ is a Lipschitz function of $\hat{\theta}({\boldsymbol{\tau}})$ with Lipschitz constant $T_{\max}$. Hence if $\left\|\hat{\theta}({\boldsymbol{\tau}}) - \hat{\theta}({\boldsymbol{\tau}_u})\right\|\leq \frac{\omega}{T_{\max}}$, then $\left|\psi(\hat{\theta}({\boldsymbol{\tau}}))-\psi(\hat{\theta}({\boldsymbol{\tau}_u}))\right|\leq \omega$ and furthermore by the Cauchy-Schwarz inequality $\left|\left\langle\hat{\theta}(\boldsymbol{\tau})-\hat{\theta}(\boldsymbol{\tau}_u),\boldsymbol{\tau}\right\rangle\right|\leq \omega$. Therefore, if $\left\|\hat{\theta}({\boldsymbol{\tau}}) - \hat{\theta}({\boldsymbol{\tau}_u})\right\|\leq \frac{\omega}{T_{\max}}$, then
\begin{align}
\left|f(\boldsymbol{\tau})-f(\boldsymbol{\tau}_u)\right| &\leq \left|\left\langle \hat{\theta}({\boldsymbol{\tau}}),\boldsymbol{\tau}\right\rangle - \left\langle \hat{\theta}(\boldsymbol{\tau}_u),\boldsymbol{\tau}_u\right\rangle\right| +\left|\psi\left(\hat{\theta}(\boldsymbol{\tau})\right)-\psi\left(\hat{\theta}(\boldsymbol{\tau}_u)\right)\right|\nonumber \\
&\leq 2\omega \label{intermediateContEq}
\end{align}
where (\ref{intermediateContEq}) follows from $\hat{\theta}(\boldsymbol{\tau}_u)=(0,\cdots,0)$, $|\langle\hat{\theta}(\boldsymbol{\tau})-\hat{\theta}(\boldsymbol{\tau}_u),\boldsymbol{\tau}\rangle|\leq \omega$. Finally, for large enough $n$ and for all $\boldsymbol{\tau}_1\in\mathcal{K}_{\epsilon}$, it holds that
\begin{align}
f(\boldsymbol{\tau}_1)&\leq \lambda+\epsilon \nonumber \\
                      &< (H+\omega)+\omega \label{stackFirst} \\
                      &= \log{|\mathcal{X}|}-3\omega. \label{lambdaLessthanEq}
\end{align}
where (\ref{stackFirst}) follows from $\lambda<H+\omega$ and the fact that since $\epsilon\rightarrow 0$, $\epsilon<\omega$ and (\ref{lambdaLessthanEq}) is from the definition of $\omega$. From (\ref{fTauUEq}) and (\ref{lambdaLessthanEq}), we have $|f(\boldsymbol{\tau}_1)-f(\boldsymbol{\tau}_u)|>2\omega$ for all $\boldsymbol{\tau}_1\in\mathcal{K}_{\epsilon}$. Hence by (\ref{intermediateContEq}), we must certainly have $\left\|\hat{\theta}({\boldsymbol{\tau}}_1) - \hat{\theta}\left({\boldsymbol{\tau}_u}\right)\right\|>\frac{\omega}{T_{\max}}$. On the other hand $\hat{\theta}({\boldsymbol{\tau}_u})=(0,\cdots,0)$, which entails that $\left\|\hat{\theta}({\boldsymbol{\tau}_1})\right\|> \frac{\omega}{T_{\max}}$. This yields a positive lower bound, independent of $n$, for the denominator in (\ref{hPrimeBound}).

\section{Proof of Lemma \ref{PointTypeClasLemma}: Point Type Class Size}
\label{app::pointProof}
The one-to-one mapping between $\boldsymbol{\tau}(x)$ and $\boldsymbol{L}(x)$, subsequently defines a one-to-one mapping between $\boldsymbol{\tau}(x^n)$ and $\boldsymbol{L}(x^n)$, which consequently defines a one-to-one correspondence between the point type class $T_{x^n}$ and $\boldsymbol{L}(x^n)$. Therefore, for any parameter value $\theta\in\Theta$, it holds that \cite{merhav}
\begin{equation}
\label{typeClassSizeApprch}
|T_{x^n}|=\frac{\mathbb{P}_{\theta}\{\boldsymbol{L}({X^n})=\boldsymbol{L}({x^n})\}}{p_{\theta}(x^n)}.
\end{equation}
Since $\boldsymbol{L}(x^n)$ can be written as a sum of integer (lattice) random vectors $\boldsymbol{L}(x_i)$ (Eq. (\ref{LSumEq})), exploiting the local limit theorem of \cite{borovkov} to bound the numerator in (\ref{typeClassSizeApprch}), yields \cite{merhav}
\begin{equation}
\log{|T_{x^n}|}=-\log{p_{\hat{\theta}(x^n)}(x^n)}-\frac{d'}{2}\log{2\pi n} -\frac{1}{2}\log{\det M\left[\hat{\theta}(x^n)\right]}+o(1)\label{weinbergerEq}
\end{equation}
where $\hat{\theta}(x^n)$ is the maximum likelihood estimate of $\theta$ for $x^n$ and $M[\theta]$ denotes the covariance matrix of the random vector $\boldsymbol{L}(X)$ where $X$ is drawn from $p_\theta$.

We show that absolute value of the third term in (\ref{weinbergerEq}), $\left|\frac{1}{2}\log{\det M\left[\hat{\theta}(x^n)\right]}\right|$, is upper bounded by a constant $C_M>0$ independent of $n$. Constant upper bound $C_{u}>0$, for $\det M\left[\hat{\theta}(x^n)\right]$ follows from Hadamard's inequality \cite[corollary 7.8.3]{horn}. For the lower bound, since $\det M[\theta]$ is a continuous function of $\theta$ over a compact domain $\Theta$, it attains a minimum at a point in the parameter space, say $\ddot{\theta}\in \Theta$. Let $\ddot{\boldsymbol{P}}$ be a diagonal $(|\mathcal{X}|-1)\times (|\mathcal{X}|-1)$ matrix with diagonal entries $\ddot{\boldsymbol{P}}_{ii}=\mathbb{P}_{\ddot{\theta}}\left(X=i+1\right)$ for $i\in\mathcal{X}$, and $\ddot{\boldsymbol{p}}$ be a column vector with $\ddot{\boldsymbol{p}}_i=\mathbb{P}_{\ddot{\theta}}\left(X=i+1\right)$ for $i\in\mathcal{X}$. We have
\begin{align}
M(\ddot{\theta})&= \mathbb{E}_{\ddot{\theta}}\left(\left[\boldsymbol{L}(X)\right]\left[\boldsymbol{L}(X)\right]^T\right)
                     - \mathbb{E}_{\ddot{\theta}}\left(\left[\boldsymbol{L}(X)\right]\right)\left(\mathbb{E}_{\ddot{\theta}}\left(\left[\boldsymbol{L}(X)\right]\right)\right)^T \nonumber \\
                &=\sum_{x\neq1}p_{\ddot{\theta}}(x)\boldsymbol{L}(x)\boldsymbol{L}(x)^T -\left(\sum_{x\neq1}p_{\ddot{\theta}}(x)\boldsymbol{L}(x)\right)\left(\sum_{x\neq1}p_{\ddot{\theta}}(x)\boldsymbol{L}(x)\right)^T \label{ZeroLone} \\
                &= \mathbb{L}\ddot{\boldsymbol{P}}\mathbb{L}^T- \left(\mathbb{L}\ddot{\boldsymbol{p}}\right)\left(\mathbb{L}\ddot{\boldsymbol{p}}\right)^T \nonumber \\
                &=\mathbb{L}(\ddot{\boldsymbol{P}}-\ddot{\boldsymbol{p}}\ddot{\boldsymbol{p}}^T)\mathbb{L}^T \label{MDecompL}
\end{align}
where (\ref{ZeroLone}) follows recalling that $\boldsymbol{L}(1)=\boldsymbol{0}$. We then show that $\left(\ddot{\boldsymbol{P}}-\ddot{\boldsymbol{p}}\ddot{\boldsymbol{p}}^T\right)$ is non-singular. Observe that
\begin{align}
\text{det}(\ddot{\boldsymbol{P}}-\ddot{\boldsymbol{p}}\ddot{\boldsymbol{p}}^T) &= (1-\boldsymbol{p}^T\ddot{\boldsymbol{P}}^{-1}\ddot{\boldsymbol{p}}) \text{det} \ddot{\boldsymbol{P}} \label{matDetLem} \\
                                                          &=\Big(1-(p_{\ddot{\theta}}(2)+\cdots+p_{\ddot{\theta}}(|\mathcal{X}|))\Big)\text{det} \ddot{\boldsymbol{P}} \nonumber \\
                                                          &=p_{\ddot{\theta}}(1) \text{det} \ddot{\boldsymbol{P}} \nonumber \\
                                                          &=p_{\ddot{\theta}}(1) p_{\ddot{\theta}}(2)\cdots p_{\ddot{\theta}}(|\mathcal{X}|) \nonumber \\
                                                          &\geq p_{\text{min}} ^{|\mathcal{X}|} \label{pminCmp}
\end{align}
where (\ref{matDetLem}) is from Matrix determinant Lemma \cite{harvill}, while existence of a constant $p_{\text{min}}$ in (\ref{pminCmp}) such that $p_{\ddot{\theta}}(x)\geq p_{\text{min}}\:\:\forall x\in\mathcal{X}$ follows from compactness of $\Theta$ and structure of the exponential family (\ref{pThetaEq}).
Since $\mathbb{L}$ is full rank and rank of a matrix is invariant under multiplication by a non-singular matrix, (\ref{MDecompL}) implies $\det M\left[\ddot{\theta}\right]>0$. Positivity of $\det M\left[\ddot{\theta}\right]$, in turn provides a positive constant lower bound $C_l$ for $\det M\left[\hat{\theta}(x^n)\right]$. Let $C_M=\frac{1}{2}\max\{|\log{C_{l}}|,|\log{C_{u}}|\}$ and $C=C_M+1$ be the constant in the lemma. Finally, lemma follows by noticing that
\begin{equation*}
\log{p_{\hat{\theta}(x^n)}(x^n)}= n\left[\left\langle \hat{\theta}(\boldsymbol{\tau}(\boldsymbol{L})),\boldsymbol{\tau}(\boldsymbol{L}) \right\rangle - \psi\left(\hat{\theta}(\boldsymbol{\tau}(\boldsymbol{L}))\right)\right]
\end{equation*}
for any $x^n$ with $\boldsymbol{L}(x^n)=\boldsymbol{L}$.

\section{Proof of Lemma \ref{disConLem}: Ratio of the Volumes}
\label{app::disConLem}
Similar to the Appendix \ref{app:fLipschitz}, one can show that $f_0(\ell)$ is a Lipschitz function of $\ell$. Therefore, for a Lipschitz constant $K_5>0$, we have $\|f_0(\ell)-f_0(\boldsymbol{L})\|\leq K_5\|\ell-\boldsymbol{L}\|$. Let $R:=\sum\limits_{i=1}^{|\mathcal{X}|}\|\boldsymbol{L}(i)\|$. We first show that
\begin{align}
\mathcal{A}_{R}&:=\left\{\ell\in\mathfrak{L}: \gamma'_0-\Delta+\frac{K_5R}{n}<f_0(\ell)\leq\gamma'_0-\frac{K_5R}{n}\right\} \subseteq \bigcup\limits_{\boldsymbol{L}\in\mathcal{A}_0}{B_{\frac{R}{n}}(\boldsymbol{L})}.\label{aaaDef}
\end{align}
For an arbitrary $\ell\in\mathcal{A}_R$, since $\mathcal{A}_R$ is a subset of the convex hull of $\mathcal{L}$, one can find real non-negative numbers $a_{i}, i=1,...,|\mathcal{X}|$ such that
\begin{equation}
\sum_{i=1}^{|\mathcal{X}|}{a_i}=1
\end{equation}
and
\begin{equation}
\ell=\sum_{i=1}^{|\mathcal{X}|}{a_i\boldsymbol{L}(i)} \label{ellFDefintiino} .
\end{equation}
For an index $j$, let $n_i=\lfloor na_i\rfloor$ for $i=1,...,j$ and $n_i=\lceil na_i\rceil$ for $i=j+1,...,|\mathcal{X}|$. We claim one can choose the index $0\leq j \leq |\mathcal{X}|$ ($j=0$ corresponds to $n_i=\lceil na_i\rceil$ for all $i$) such that $\sum_{i=1}^{|\mathcal{X}|}n_i=n$. Observe that for $j=0$, we have $\sum_{i=1}^{|\mathcal{X}|}n_i\geq n$, while for $j=|\mathcal{X}|$, $\sum_{i=1}^{|\mathcal{X}|}n_i\leq n$. Incrementing $j$ by one, decreases the integer $\sum_{i=1}^{|\mathcal{X}|}n_i$ by at most one. The claim then follows.

It is clear that $n_i$'s satisfy the following condition as well
\begin{equation}
\label{oneLeftLast}
|n_i-na_i|< 1, \:\:\forall i=1,...,|\mathcal{X}|.
\end{equation}
Let $x^n\in\mathcal{X}^n$ be any sequence with empirical probability mass function $\left\{\frac{n_i}{n}\right\}$. Observe that
\begin{equation*}
\boldsymbol{L}(x^n)=\frac{1}{n}\sum_{i=1}^{|\mathcal{X}|}n_i\boldsymbol{L}(i)\in\mathcal{L}.
\end{equation*}
Therefore one obtains
\begin{align}
\|\ell-\boldsymbol{L}(x^n)\|  &\leq \frac{1}{n}\sum_{i=1}^{|\mathcal{X}|}\left|n_i-na_i\right|\cdot \left\|\boldsymbol{L}(i)\right\| \label{cauchyEq} \\
                         &< \frac{1}{n}\sum_{i=1}^{|\mathcal{X}|}\left\|\boldsymbol{L}(i)\right\| \label{consEq} \\
                         &= \frac{R}{n} \label{yeksad}
\end{align}
where (\ref{cauchyEq}) follows from (\ref{ellFDefintiino}) and the Cauchy-Schwarz inequality, (\ref{consEq}) follows from (\ref{oneLeftLast}). Therefore $\ell\in B_{\frac{R}{n}}(\boldsymbol{L}(x^n))$. We then show that $\boldsymbol{L}(x^n)\in\mathcal{A}_0$. From (\ref{yeksad}) and the Lipschitzness of $f_0(\cdot)$ we have
\begin{equation}
\label{togewtherEq}
f_0(\ell)-\frac{K_5R}{n}\leq f_0\left(\boldsymbol{L}(x^n)\right)\leq f_0(\ell)+\frac{K_5R}{n}.
\end{equation}
From (\ref{togewtherEq},\ref{aaaDef}) and since $\ell\in\mathcal{A}_R$, we obtain
\begin{equation}
\gamma'_0-\Delta< f_0\left(\boldsymbol{L}(x^n)\right)\leq \gamma'_0
\end{equation}
which confirms $\boldsymbol{L}(x^n)\in\mathcal{A}_0$. Since for an arbitrary $\ell\in\mathcal{A}_R$, we are able to find $\boldsymbol{L}(x^n)\in\mathcal{A}_0$ within a distance of $\frac{R}{n}$, (\ref{aaaDef}) follows.

We continue by observing the following
\begin{align}
\text{Vol}\left(\bigcup_{\ell\in\mathcal{A}_R}B_{\frac{1}{2n}}(\ell)\right)&\leq \text{Vol} \left(\bigcup_{\boldsymbol{L}\in\mathcal{A}_0}B_{\frac{2R+1}{2n}}(\boldsymbol{L})\right) \label{ARSIcVol} \\
                                                               &\leq (2R+1)^{d'}\text{Vol} \left(\bigcup_{\boldsymbol{L}\in\mathcal{A}_0}B_{\frac{1}{2n}}(\boldsymbol{L})\right) \label{multPFac}
\end{align}
where (\ref{ARSIcVol}) is from (\ref{aaaDef}) and a geometrical observation (triangle inequality) that if a point is within a distance $\frac{1}{2n}$ of a point in $\mathcal{A}_R$, it is certainly within a distance $\frac{R}{n}+\frac{1}{2n}$ of a point in $\mathcal{A}_0$, (\ref{multPFac}) follows since scaling the radius of an sphere by a constant, changes its volume by a constant multiplicative factor.

Given (\ref{multPFac}), to prove the lemma it is enough to show that for some constant $C>0$,
\begin{equation}
\frac{\text{Vol}\left(\bigcup_{\ell\in\tilde{\mathcal{A}}_0}B_{\frac{1}{2n}}(\ell)\right)}{\text{Vol}\left(\bigcup_{\ell\in\mathcal{A}_R}B_{\frac{1}{2n}}(\ell)\right)} \leq C. \label{thispneRep}
\end{equation}

Observe the following
\begin{align}
&\text{Vol}\left(\bigcup_{\ell\in\tilde{\mathcal{A}}_0}B_{\frac{1}{2n}}(\ell)\right)- \text{Vol}\left(\bigcup_{\ell\in\mathcal{A}_R}B_{\frac{1}{2n}}(\ell)\right)  \label{firstEq} \\
&\leq\text{Vol}\left(\ell: f_0(\ell)\in \left(\gamma'_0-\Delta-\frac{K_5}{2n},\gamma'_0+\frac{K_5}{2n}\right]\right)\label{fnotConv}\\
&\hspace{0.25in}-\text{Vol}\left(\ell: f_0(\ell)\in \left(\gamma'_0-\Delta+\frac{K_5R}{2n},\gamma'_0-\frac{K_5R}{2n}\right]\right) \label{fNotAR} \\
&=\rho_0\left(\gamma'_0+\frac{K_5}{2n}\right)-\rho_0\left(\gamma'_0-\Delta-\frac{K_5}{2n}\right)+\rho_0\left(\gamma'_0-\Delta+\frac{K_5R}{2n}\right)-\rho_0\left(\gamma'_0-\frac{K_5R}{2n}\right) \label{rhonitDef}\\
&\leq \frac{C}{n} \label{finaleEw}
\end{align}
where (\ref{fnotConv}) is an upper bound for the first term in (\ref{firstEq}) noticing the definition of $\tilde{\mathcal{A}}_0$ and Lipschitzness of $f_0(\cdot)$, (\ref{fNotAR}) is from lower bounding the volume of the ball-covering of $\mathcal{A}_R$ (second term in (\ref{firstEq})) by the volume of $\mathcal{A}_R$ itself, (\ref{rhonitDef}) is from the definition of $\rho_0(\cdot)$, and (\ref{finaleEw}) is from Lipschitzness of $\rho_0(\cdot)$ and recalling the choice of $\Delta=\frac{1}{n}$.
Therefore
\begin{align}
\frac{\text{Vol}\left(\bigcup_{\ell\in\tilde{\mathcal{A}}_0}B_{\frac{1}{2n}(\ell)}\right)}{\text{Vol}\left(\bigcup_{\ell\in\mathcal{A}_R}B_{\frac{1}{2n}}(\ell)\right)}
&\leq \frac{\text{Vol}\left(\bigcup_{\ell\in\mathcal{A}_R}B_{\frac{1}{2n}}(\ell)\right)+\frac{C}{n}}{\text{Vol}\left(\bigcup_{\ell\in\mathcal{A}_R}B_{\frac{1}{2n}}(\ell)\right)} \nonumber \\
&=1+\frac{C}{n\text{Vol}\left(\bigcup_{\ell\in\mathcal{A}_R}B_{\frac{1}{2n}}(\ell)\right)} \nonumber \\
&\leq 1+\frac{C}{n\left(\rho_0\left(\gamma'_0-\frac{K_5R}{2n}\right)-\rho_0\left(\gamma'_0-\Delta+\frac{K_5R}{2n}\right)\right)} \label{bothCovIt} \\
&\leq 1+\frac{C}{K_4(K_5R+1)} \label{inAkhar}
\end{align}
where (\ref{bothCovIt}) is by lower bounding the volume of the ball-covering of $\mathcal{A}_R$ by the volume of $\mathcal{A}_R$ itself, along with the definition of $\rho_0(\cdot)$, and (\ref{inAkhar}) is an application of Lemma \ref{rho0LowerDrive} as well as recalling the choice of $\Delta=\frac{1}{n}$. This proves (\ref{thispneRep}), and the lemma follows.

\section{Proof of Lemma \ref{rho0LowerDrive}: Lower bound on $|\frac{d}{d\lambda}\rho_0(\lambda)|$}
\label{sec::rhoZeroAppx}
Denote $\mathcal{K}_0=\{\ell\in\mathfrak{L}: f_0(\ell)\leq \lambda\}$ and $\mathcal{K}_0^c=\mathfrak{L}\backslash\mathcal{K}_0$. Furthermore, let us denote $\mathcal{K}_{0,\epsilon}=\{\ell\in\mathfrak{L}:f_0(\ell)\leq \lambda+\epsilon\}$. We have
\begin{align}
\frac{d}{d\lambda}\rho_0(\lambda) &= \lim_{\epsilon\rightarrow 0}\frac{\rho_0(\lambda+\epsilon)-\rho_0(\lambda)}{\epsilon} \nonumber \\
                                  &=\lim_{\epsilon\rightarrow0}\frac{\text{Vol}(\mathcal{K}_{0,\epsilon})-\text{Vol}(\mathcal{K}_0)}{\epsilon}.\label{rhonotPrime}
\end{align}
Let $\ell_1$ be an arbitrary point in $\mathcal{K}_0$. Let
\begin{equation}
\label{secondVicinityEq}
\ell_2=\ell_1+\frac{\epsilon}{2}\frac{\nabla f_0(\ell_1)}{\|\nabla f_0(\ell_1)\|^2}.
\end{equation}
From Taylor series expansion, we have
\begin{equation}
\label{ellTaylor}
f_0(\ell_2)=f_0(\ell_1)+\left\langle\nabla f_0(\ell_1),\ell_2-\ell_1\right\rangle + \Delta_0
\end{equation}
where $|\Delta_0|\leq C_{f_0}\|\ell_1-\ell_2\|^2$, for a constant $C_{f_0}$ which is independent of $n$.
First observe from (\ref{secondVicinityEq}) that, since $\epsilon\rightarrow0$ is infinitesimal, $\ell_2$ resides in the vicinity of $\ell_1$ with distance at most
\begin{equation}
\label{vicinityEq}
\|\ell_2-\ell_1\|< \sqrt{\frac{\epsilon}{2C_{f_0}}}.
\end{equation}
With the choice of $\ell_2$ in (\ref{secondVicinityEq}), we have
\begin{align}
f_0(\ell_2)&<f_0(\ell_1)+\frac{\epsilon}{2}+\frac{\epsilon}{2} \label{firtELL} \\
                   &\leq \lambda+\epsilon \label{secondEll}
\end{align}
where (\ref{firtELL}) follows from (\ref{secondVicinityEq},\ref{ellTaylor},\ref{vicinityEq}), and (\ref{secondEll}) is a consequence of $\ell_1$ being a point in $\mathcal{K}_0$.
Therefore $\ell_2\in\mathcal{K}_{0,\epsilon}$. As a conclusion for all $\ell_1\in\mathcal{K}_0$, $\ell_1+\frac{\epsilon}{2}\frac{\nabla f_0(\ell_1)}{\|f_0(\ell_1)\|^2}\in \mathcal{K}_{0,\epsilon}$. That translates into the following subset relationship
\begin{equation}
\label{seubsetRelEq}
\mathcal{K}_0+B\left(\frac{\epsilon}{2}\frac{\nabla f_0(\ell_1)}{\|f_0(\ell_1)\|^2}\right)\subset \mathcal{K}_{0,\epsilon}.
\end{equation}
Continuing from (\ref{rhonotPrime}) we have
\begin{align}
\left|\frac{d}{d\lambda}\rho_0(\lambda)\right|&\geq  \lim_{\epsilon\rightarrow 0}\frac{\text{Vol}\left(\mathcal{K}_0+B\left(\frac{\epsilon}{2}\frac{\nabla f_0(\ell_1)}{\|\nabla f_0(\ell_1)\|^2}\right)\right)-\text{Vol}(\mathcal{K}_0)}{\epsilon} \label{yek} \\
              &\geq \frac{S(\mathcal{K}_0)}{2\|\nabla f_0(\ell_1)\|}  \label{do}\\
              &= \frac{S(\mathcal{K}_0)}{2\|\hat{\theta}(\ell_1)\|} \label{se} \\
              &\geq \frac{S(\mathcal{K}_0)}{\wp} \label{char}
\end{align}
where (\ref{yek}) is a consequence of (\ref{seubsetRelEq}), (\ref{do}) is due to the definition of the surface area in (\ref{surfaceDef}), (\ref{se}) is derived similar to (\ref{ssstackRellb}), and finally (\ref{char}) is from the fact that for all $\theta\in\Theta$, we have $\|\theta\|\leq \wp$.

It remains to provide a positive constant lower bound for $S(\mathcal{K}_0)$ independent of $n$. We first show that in the range $\gamma'_0-\Delta\leq\lambda\leq\gamma'_0$, there exists a positive constant lower bound for $\text{Vol}(\mathcal{K}_0)$. Since
 \begin{align}
 \Upsilon(\ell)&:= -\frac{1}{n}\left(\left\langle\hat{\theta}(\boldsymbol{\tau}(\ell)),\boldsymbol{\tau}(\ell)\right\rangle-\psi\left(\hat{\theta}\left(\boldsymbol{\tau}(\ell)\right)\right)\right) \label{fNotContEqI} \\
 &=-\frac{1}{n} \left(\left\langle \hat{\theta}(\boldsymbol{\mathbf{b}}+\boldsymbol{\mathbb{A}} \boldsymbol{R} \ell),\boldsymbol{\mathbf{b}}+\boldsymbol{\mathbb{A}}\boldsymbol{R}\ell \right\rangle - \psi\left(\hat{\theta}(\boldsymbol{\mathbf{b}}+\boldsymbol{\mathbb{A}}\boldsymbol{R}\ell)\right)\right) \label{fNotContEqII}
 \end{align}
 is a continuous function of $\ell$ over a compact domain $\mathfrak{L}$, it attains a minimum at a point, say $\ell^*\in \mathfrak{L}$. This minimum is certainly less than or equal to the minimum of $\Upsilon(\ell)$ over $\mathcal{L}$, which is attained at a point say $\boldsymbol{L}^*$. For any ${\theta}\in\Theta$, we have
\begin{align}
\Upsilon(\ell^*) &\leq \Upsilon(\boldsymbol{L}^*) \nonumber \\
                    &\leq \sum_{x^n}{p_{\theta}(x^n)}\left(-\frac{1}{n}\log{p_{\hat{\theta}(x^n)}(x^n)}\right) \label{minExp}\\
                                                   &\leq \sum_{x^n}{p_{\theta}(x^n)}\left(-\frac{1}{n}\log{p_{\theta}(x^n)}\right) \label{maxTrueRe} \\
                                                   &=H(p_{\theta}) \label{iHaveto}
\end{align}
where (\ref{minExp}) follows since $\Upsilon(\boldsymbol{L}^*)=-\frac{1}{n}\log{p_{\hat{\theta}(y^n)}(y^n)}$ for some $y^n\in\mathcal{X}^n$ with $\boldsymbol{L}(y^n)=\boldsymbol{L}^*$, more precisely
\begin{equation*}
\Upsilon(\boldsymbol{L}^*)=\min_{x^n\in\mathcal{X}^n}{-\frac{1}{n}\log{p_{\hat{\theta}(x^n)}(x^n)}}
\end{equation*}
 and the minimum value of a function is less than or equal to its weighted average with respect to any weighting, (\ref{maxTrueRe}) is from $p_{\hat{\theta}(x^n)}(x^n)\geq p_{\theta}(x^n)$.

 Recall $H:=H(p_{\theta^*})$ as the entropy of the underlying model. We provide a positive lower bound, independent of $n$ for $\delta$ defined as follows:
\begin{equation}
\label{deltaAppros}
\delta :=H-\Upsilon(\ell^*).
\end{equation}
We assume that the underlying model is not the lowest entropy model in the class, i.e. $H>\min_{\theta\in\Theta}H(p_{\theta})$. Since $H(p_{\theta})$ is a continuous function of $\theta$ over a compact domain, $\min_{\theta\in\Theta}H(p_{\theta})$ is achieved for a model in the class, say $\theta_{min}\in\Theta$. We then have
\begin{align}
\delta &\geq H-H(p_{\theta_{min}}) \label{minHandtheta} \\
       &>0 \label{trueHandminH}
\end{align}
where (\ref{minHandtheta}) follows from (\ref{deltaAppros}) and since (\ref{iHaveto}) is true for any ${\theta}\in\Theta$ including $\theta_{min}$, (\ref{trueHandminH}) is from the assumption that $H>\min_{\theta\in\Theta}H(p_{\theta})$.

Similar to the Appendix \ref{app:fLipschitz}, one can show that $f_0(\ell)$ is a Lipschitz function of $\ell$ with Lipschitz constant $K_5>0$. For any $\ell\in\mathfrak{L}$ with $\|\ell-\ell^*\|\leq \frac{\delta}{2K_5}$, we have

\begin{align}
f_0(\ell)&\leq f_0(\ell^*)+K_5\cdot\frac{\delta}{2K_5} \label{f0Lipschits} \\
         &= \Upsilon(\ell^*)-\frac{d'}{2n}\log{(2\pi n)} +\frac{C}{n} + \frac{\delta}{2} \label{fNotDefinition} \\
         &= H-\delta -\frac{d'}{2n}\log{(2\pi n)} +\frac{C}{n} + \frac{\delta}{2} \label{ConsAbovEq} \\
         &< H-\frac{\delta}{3} \label{largeEnough}\\
         &<\gamma'_0 -\Delta \label{largeEnoughSecond} \\
         &\leq \lambda \label{rangeOfLambda}
\end{align}
where (\ref{f0Lipschits}) follows from the Lipschitzness of $f_0(\cdot)$ with Lipschitz constant $K_5$, (\ref{fNotDefinition}) is from (\ref{fNotContEqI},\ref{subnserf}), (\ref{ConsAbovEq}) is from the definition of $\delta$ in (\ref{deltaAppros}), (\ref{largeEnough}) holds for large enough $n$, and (\ref{largeEnoughSecond}) holds for large enough $n$, recalling the choices of $\gamma'_0$ in (\ref{gammaPrimeZeroEq}) and $\Delta=\frac{1}{n}$, and (\ref{rangeOfLambda}) is due to the range of $\lambda$. Therefore, from the definition of $\mathcal{K}_0$, we obtain the following relation
\begin{equation*}
\left\{\ell\in\mathfrak{L}:\|\ell-\ell^*\|\leq \frac{\delta}{2K_5}\right\}\subset \mathcal{K}_0.
\end{equation*}
Hence
\begin{align}
\text{Vol}(\mathcal{K}_0)&\geq \text{Vol}\left(\left\{\ell\in\mathfrak{L}:\|\ell-\ell^*\|\leq\frac{\delta}{2K_5}\right\}\right) \nonumber \\
                         &=C\left(\frac{\delta}{2K_5}\right)^{d'} \label{sphereSection} \\
                         &\geq C\left(\frac{H-H(p_{\theta_{min}})}{2K_5}\right)^{d'} \label{insertDFelra}
\end{align}
where (\ref{sphereSection}) is from the fact that the intersection of the sphere $\|\ell-\ell^*\|\leq \frac{\delta}{2K_5}$ and $\mathfrak{L}$ is independent of $n$ and only depends on the constellation of $\mathcal{L}$, and (\ref{insertDFelra}) is from (\ref{minHandtheta}).

Finally, since sphere has the smallest surface area among all shapes of a given volume, therefore a positive constant lower bound on $\text{Vol}(\mathcal{K}_0)$, implies a positive constant lower bound on $S(\mathcal{K}_0)$. More precisely, recalling the equations for the volume and the surface area of a $d'$-dimensional sphere \cite[Eq. 1.5.1]{ren}, we have
\begin{equation*}
S(\mathcal{K}_0)\geq C\left(\frac{H-H(p_{\theta_{min}})}{2K_5}\right)^{d'} \frac{2\sqrt{\pi}\Gamma\left(\frac{d'}{2}+1\right)}{\Gamma\left(\frac{d'+1}{2}\right)}.
\end{equation*}

\end{appendices}

{\bibliographystyle{IEEEtran}
\bibliography{reputation}}
\end{document}